\documentclass[letterpaper, twocolumn]{IEEEtran}

\usepackage[style=ieee, giveninits=true, maxnames=99]{biblatex}
\usepackage{amsmath}
\usepackage{physics}
\usepackage{xargs}
\usepackage{xifthen}
\usepackage{dsfont}
\usepackage{mathtools}

\usepackage[amsmath,standard,thmmarks,hyperref]{ntheorem}

\theoremclass{theorem}
\theoremseparator{.}
\renewtheorem{theorem}{Theorem}
\renewtheorem{lemma}[theorem]{Lemma}
\renewtheorem{corollary}[theorem]{Corollary}
\renewtheorem{proposition}[theorem]{Proposition}

\theoremclass{definition}
\theoremseparator{.}
\renewtheorem{definition}[theorem]{Definition}

\newtheorem{assumption}[theorem]{Assumption}

\theoremclass{example}
\theoremseparator{.}
\renewtheorem{example}[theorem]{Example}

\theoremclass{proof}
\theoremseparator{:}
\renewtheorem{proof}{Proof}

\usepackage{graphicx}
\usepackage{microtype}
\graphicspath{{./images/}}
\usepackage[caption=false,font=footnotesize]{subfig}

\usepackage[text=cicero]{lipsum}

\usepackage[dvipsnames]{xcolor}
\usepackage[unicode,bookmarks,colorlinks,breaklinks]{hyperref}
\hypersetup{
    linkcolor=NavyBlue,
    citecolor=OliveGreen,
    filecolor=Black,
    urlcolor=NavyBlue
}
\usepackage[nameinlink,capitalize]{cleveref}
\crefname{assumption}{Assumption}{Assumptions}
\Crefname{assumption}{Assumption}{Assumptions}

\usepackage[shortlabels]{enumitem}

\usepackage[most]{tcolorbox}

\usepackage{tikz}
\usetikzlibrary{calc}
\usetikzlibrary{positioning}
\usetikzlibrary{arrows.meta}
\usetikzlibrary{fit}

\usepackage{silence}
\DeclarePairedDelimiter\aparen{\lparen}{\rparen}
\DeclarePairedDelimiter\abrack{\lbrack}{\rbrack}
\DeclarePairedDelimiter\abrace{\lbrace}{\rbrace}

\DeclarePairedDelimiter\afloor{\lfloor}{\rfloor}

\renewcommand{\pqty}[1]{\aparen*{#1}}
\renewcommand{\bqty}[1]{\abrack*{#1}}
\renewcommand{\Bqty}[1]{\abrace*{#1}}

\newcommand{\floor}[1]{\afloor*{#1}}

\newcommand{\up}{\mathrm}
\newcommand{\ca}{\mathcal}

\newcommand{\bb}{\mathbb}

\newcommand{\R}{\bb{R}}

\newcommand{\N}{\bb{N}}
\newcommand{\Z}{\bb{Z}}

\newcommand{\upe}{\up{e}}

\newcommand{\gib}{\,|\,}

\DeclareMathOperator*{\Prb}{\bb{P}}
\renewcommand{\Pr}[1]{\Prb\pqty{#1}}

\DeclareMathOperator*{\E}{\bb{E}}
\newcommand{\Exp}[2][]{\ifthenelse{\isempty{#1}}{\E\bqty{#2}}
    {\E_{#1}\bqty{#2}}}

\newcommand{\dilog}[1]{\operatorname{Li_2}\pqty{#1}}

\DeclareMathOperator*{\argmax}{\arg\max}

\newcommand{\imgwidth}{0.4\linewidth}
\newcommand{\wT}{w_{\tau}}
\newcommand{\wH}{w_{\up{H}}}
\newcommand{\GT}{G_{\tau}}
\newcommand{\GH}{G_{\up{H}}}
\newcommand{\bsc}{\textsf{BSC}}

\title{The Linear Reliability Channel}

\author{Alexander Mariona, Ken R.~Duffy, Muriel M\'{e}dard}
\author{
\IEEEauthorblockN{%
Alexander Mariona\IEEEauthorrefmark{1}, %
Ken R.~Duffy\IEEEauthorrefmark{2}, and %
Muriel M\'{e}dard\IEEEauthorrefmark{1}}

\IEEEauthorblockA{\IEEEauthorrefmark{1}%
Research Laboratory of Electronics, %
Massachusetts Institute of Technology, %
Cambridge, MA, USA}

\IEEEauthorblockA{\IEEEauthorrefmark{2}%
Dept.~of Mathematics and Dept.~of ECE, %
Northeastern University, %
Boston, MA, USA\\
E-mail: amariona@mit.edu, k.duffy@northeastern.edu, medard@mit.edu}
}

\begin{document}

\maketitle

\begin{abstract}
We introduce and analyze a discrete soft-decision channel called the linear
reliability channel (LRC) in which the soft information is the rank ordering of
the received symbol reliabilities. We prove that the LRC is an appropriate
approximation to a general class of discrete modulation, continuous noise
channels when the noise variance is high. The central feature of the LRC is
that its combinatorial nature allows for an extensive mathematical analysis of
the channel and its corresponding hard- and soft-decision maximum likelihood
(ML) decoders. In particular, we establish explicit error exponents for ML
decoding in the LRC when using random codes under both hard- and soft-decision
decoding. This analysis allows for a direct, quantitative evaluation of the
relative advantage of soft-decision decoding. The discrete geometry of the LRC
is distinct from that of the BSC, which is characterized by the Hamming weight,
offering a new perspective on code construction for soft-decision settings.
\end{abstract}


\section{Introduction}

Error correction decoding algorithms are broadly divisible into hard-decision
and soft-decision decoders \cite{Gal08}. Hard-decision decoders are algorithms
that take as input only bits, whereas soft-decision decoders also make use of
side information, referred to as soft information, quantifying the likelihood
that each bit is correct. The standard form of soft information per bit is the
log-likelihood ratio (LLR) of the hypotheses that the transmitted bit is 0 or 1
given the channel output.

Ordered Reliability Bits Guessing Random Additive Noise Decoding (ORBGRAND) is
a code-agnostic, soft-decision decoding algorithm \cite{DWM22} that has
recently been shown to be almost capacity-achieving for the real-valued
additive white Gaussian noise channel \cite{Liu+22} and to be practically
feasible via efficient hardware implementation, both in synthesis
\cite{Abb+22,Con22,Ji+25,Xia+23} and silicon \cite{Ria+25}. Subsequent
theoretical work has explored algorithmic modifications to approach the
performance of ML soft-decision decoding while maintaining efficiency
\cite{WYZ25,WZ24} and has studied the achievable rate of ORBGRAND in more
general settings \cite{LZ24}.

Motivated by these developments, this work formalizes the fundamental
algorithmic insight of ORBGRAND, the approximation of the sorted magnitudes of
the received LLRs by a linear function, into a channel model for which this
linear behavior is exact. For this channel, which we call the linear
reliability channel (LRC), ORBGRAND is a true maximum-likelihood (ML)
soft-decision decoder. A key feature of the LRC is that it is a discrete
soft-decision channel in which the soft information is combinatorial and
sufficiently structured to allow for a complete mathematical analysis of the
maximum-likelihood decoding, both hard- and soft-decision. The behavior of the
LRC is aligned with a general family of continuous-noise channels at low
signal-to-noise ratios, where decoding performance is most relevant. In the
LRC, the received bit reliabilities, i.e., the magnitudes of the LLRs of the
received bits, are linearly increasing when subject to a random permutation.
The soft information is, therefore, the permutation for a given channel use,
and the knowledge of that permutation suffices for an exact ML decoding.
Intrinsically connected with the LRC and its ML decoder is a statistic called
the logistic weight, which is analogous to the Hamming weight in the context of
the BSC. The noise level in the LRC is parameterized by the slope of the linear
increase in reliabilities, and this slope plays an analogous role to the
bit-flip probability in a BSC. When the slope is large, most bit are
transmitted reliability, whereas a significant portion are unreliable when the
slope is small.

We derive closed-form, computable error exponents for both hard- and
soft-decision ML decoding in the LRC. In order to do so, we leverage the
mathematical framework of large deviations and guesswork, as introduced in
\cite{Ari00,MS04,PS04,HS10,CD12}. At a high level, we show that the
guesswork process for the noise in the LRC satisfies a large deviation
principle (LDP) in both the hard- and soft-decision settings. Having
established these LDPs, we utilize the formulation of the channel coding
theorem presented in \cite{DLM19}, which results in explicit expressions for
the error and success exponents, under the assumption that the code book is
chosen uniformly at random. These exponents show that, in the large block
length limit, soft-decision decoding strictly outperforms hard-decision
decoding in the LRC. This analysis allows for a quantitative evaluation of the
performance difference between hard- and soft-decision decoding at any code
rate and any noise level.


\section{Overview of Results}

We present here an outline of the sequel, summarizing the main results and
offering intuitive interpretations of the more technical statements.

\Cref{sec:model} defines the LRC and presents its key properties. We show in
\cref{sec:approx} how the LRC can be viewed as an approximation to binary-input
channels with independent additive noise described by a symmetric, strictly
log-concave, and sufficiently smooth ``location-scale'' distribution at lower
signal-to-noise ratios (\cref{thm:lrc-approx}). Examples of such distributions
include the normal, logistic, and Laplace distributions. This approximation
justifies the linear reliability phenomenon as being a suitable foundation of a
general framework for soft-decision decoding. A key consequence of linear
reliabilities is that the logistic weight (\cref{def:lw}) of a binary sequence
is the characteristic statistic for soft-decision decoding, in the same way
that the Hamming way is characteristic for hard-decision decoding.
\Cref{sec:weight} catalogues some basic properties of the logistic weight and
presents a recent number theoretic result due to Bridges \cite{Bri20} that
allows us to determine an accurate approximation to the number of sequences of
length $n$ and logistic weight $w$ (\cref{thm:bridges}). This approximation is
used in a manner akin to Stirling's approximation to the binomial coefficient.

\Cref{sec:ml} introduces the soft-decision (\cref{thm:ml-soft}) and
hard-decision ML decoders (\cref{thm:ml-hard}) for the LRC. We analyze these
algorithms by interpreting them as executions of Guessing Random Additive Noise
Decoding (GRAND), a family of code-agnostic channel decoding algorithms
\cite{DLM19,DWM22} based on the information theoretic concept of guesswork
\cite{Mas94,Ari96,MS04}. This viewpoint allows us to leverage a unified
mathematical framework, the theory of large deviations, for the analysis of the
probability of a decoding error. A secondary benefit is that this perspective
clearly highlights the role of the logistic weight in soft-decision decoding
and that of the Hamming weight in hard-decision decoding for the LRC.

\Cref{sec:ldp} establishes large deviation principles (LDPs) for the exponent
of the number of guesses made by the soft-decision
(\cref{thm:scgf-soft,thm:rate}) and hard-decision ML decoders
(\cref{thm:scgf-hard,thm:rate}) in the LRC. An LDP is the key analytical tool
in large deviations theory \cite{Tou09,DZ10,Var84,DS01}, and the techniques we
employ to establish these LDPs developments from work on the large deviations
of guesswork \cite{Ari00,HS10,CD12}. For our purposes, an LDP can be
intuitively understood as quantitatively describing the exponential decay of
the probability that a sequence of random variables has a realization which is
a ``large deviation'' from its typical value. This decay rate is the asymptotic
limiting rate as the parameter value tends towards infinity. The number of
guesses a decoder makes is closely related to the probability that the decoding
is correct when using a random code. Thus, the LDPs proven in this section
describe the key properties of the ML decoders for understanding their decoding
behavior. In the process of proving the guesswork LDPs, we establish a key
property of the noise distribution in the LRC, namely, that the R\'enyi entropy
of the noise is always lower after conditioning on the soft information
(\cref{lm:renyi-order}). This is shown to imply that the capacity of the LRC is
strictly higher under soft-decision decoding compared to hard-decision
decoding.

\Cref{sec:error} leverages those LDPs to establish error exponents, for the
probability of incorrectly decoding below capacity, and success exponents, for
the probability of correctly decoding above capacity, for both the soft- and
hard-decision ML decoders. These follow from the large-deviations channel
coding theorem for random codes as formulated in \cite{DLM19}, in contrast to
the direct techniques for discrete memoryless channels dating back to Shannon,
Gallager, and Berlekamp \cite{Gal68,SGB67,Ber02}. In addition to providing both
error and success exponents together, the large-deviations approach leads to a
natural interpretation of the critical rate, the point at which the error
exponent transitions from being linear to strictly convex. First observed by
Gallager \cite{Gal65}, this phenomenon lacked an intuitive interpretation in
terms of actual decoder behavior. In the context of the LRC, we show that the
critical rate always occurs earlier for hard-decision decoding than for
soft-decision decoding (\cref{pr:rcr}). Combined with the ordering of R\'{e}nyi
entropies, this suffices to prove that the error and success exponents are
always better under soft-decision decoding in the LRC (\cref{pr:soft-win}). By
comparing the error exponents for the LRC to those for the BSC, we also give a
heuristic interpretation of how ``noisy'' the LRC is at particular noise
parameter values (\cref{fig:bsc-lrc}). Roughly speaking, when the slope of the
reliabilities in the LRC is on the order of $10^x$, the hard-decision error
exponent is comparable to the BSC error exponent for $p=10^{-(x+1)}$.

We provide concluding thoughts in \cref{sec:conclusion}, followed by two
appendices that contain proofs that are deferred due to their length.
\Cref{apx:scgf-hard} details the proof of \cref{thm:scgf-hard}, relating to the
LDP for hard-decision guesswork. \Cref{apx:d-order} states and proves
\cref{lm:scgf-d-order}, which shows that the critical rate is lower under
hard-decision decoding.




\section{The Channel Model}\label{sec:model}

Throughout, we use the following notational conventions. The natural logarithm
is denoted by $\ln$ and the base-2 logarithm is denoted by $\log_2$. The set of
integers from 1 to $n$ is denoted by $[n]$. The set of all permutations of
$[n]$ is denoted by $S_n$. The Hamming weight of a binary sequence $x$ is
denoted by $\wH(x)$. Probability mass functions (PMFs) of discrete random
variables are denoted by lowercase $p$ and probability density functions (PDFs)
of continuous random variables are denoted by lowercase $f$.

\begin{figure}
\centering
\begin{tikzpicture}[node distance=2cm anchor=center, auto]
    \tikzset{myn/.style={draw, rectangle, align=center, minimum width=2.25cm}}
    \tikzset{big arrow/.style={-{Latex[scale=1]}}}
    \newcommand{\vs}{-1.1}
    \newcommand{\bscstep}{1}

    \node (X1)    at (0, 0*\vs) {$X_1$};
    \node (X2)    at (0, 1*\vs) {$X_2$};
    \node[fit={(X1.west) (X1.east)}, inner sep=0pt]
        (dots1) at (0, 2*\vs) {$\vdots$};
    \node (Xn)    at (0, 3*\vs) {$X_n$};

    \node (tauX1) at (1.5, 0*\vs) {$T(X)_1$};
    \node (tauX2) at (1.5, 1*\vs) {$T(X)_2$};
    \node[fit={(tauX1.west) (tauX1.east)}, inner sep=0pt]
        (dots2) at (1.5, 2*\vs) {$\vdots$};
    \node (tauXn) at (1.5, 3*\vs) {$T(X)_n$};

    \node[myn] (BSC1) at (4, 0*\vs) {\bsc$(q_1)$ \\ $\up{Rel}=\beta/n$};
    \node[myn] (BSC2) at (4, 1*\vs) {\bsc$(q_2)$ \\ $\up{Rel}=2\beta/n$};
    \node[myn] (BSCn) at (4, 3*\vs) {\bsc$(q_n)$ \\ $\up{Rel}=\beta$};
    \node     (dots3) at (4, 2*\vs) {$\vdots$};

    \node (tauY1) at (6.5, 0*\vs) {$T(Y)_1$};
    \node (tauY2) at (6.5, 1*\vs) {$T(Y)_2$};
    \node[fit={(tauY1.west) (tauY1.east)}, inner sep=0pt]
        (dots4) at (6.5, 2*\vs) {$\vdots$};
    \node (tauYn) at (6.5, 3*\vs) {$T(Y)_n$};

    \node (Y1)    at (8, 0*\vs) {$Y_1$};
    \node (Y2)    at (8, 1*\vs) {$Y_2$};
    \node[fit={(Y1.west) (Y1.east)}, inner sep=0pt]
        (dots5) at (8, 2*\vs) {$\vdots$};
    \node (Yn)    at (8, 3*\vs) {$Y_n$};

    \draw[big arrow] (X1.east)    -- (tauX2.west);
    \draw[big arrow] (X2.east)    -- (tauXn.west);
    \draw[big arrow] (Xn.east)    -- (tauX1.west);

    \draw[big arrow] (tauX1) -- (BSC1);
    \draw[big arrow] (tauX2) -- (BSC2);
    \draw[big arrow] (tauXn) -- (BSCn);

    \draw[big arrow] (BSC1) -- (tauY1);
    \draw[big arrow] (BSC2) -- (tauY2);
    \draw[big arrow] (BSCn) -- (tauYn);

    \draw[big arrow] (tauY1.east) -- (Yn.west);
    \draw[big arrow] (tauY2.east) -- (Y1.west);
    \draw[big arrow] (tauYn.east) -- (Y2.west);

    \node at ($ (X1.east)!0.5!(tauX1.west) + (0, 0.5)$) {$T$};
    \node at ($ (tauY1.east)!0.5!(Y1.west) + (0, 0.5)$) {$T^{-1}$};
\end{tikzpicture}

\caption{The linear reliability channel. The reliability ordering permutation
$T\in S_n$ is chosen uniformly at random with each channel use. The input bit
$X_i$ is received through a BSC with bit-flip probability $q_{T(i)}$
(\cref{eq:lrc-q}), resulting in a reliability of $T(i)\beta/n$.}
\label{fig:lrc}
\end{figure}
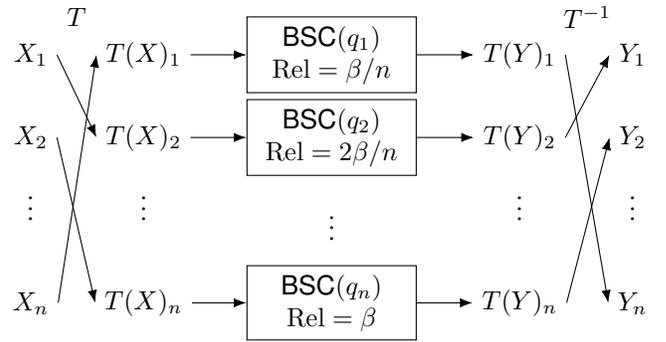

The fundamental and defining property of the LRC is the fact that the
magnitudes of the log-likelihood ratios, which we refer to as the
\emph{reliabilities}, of the received symbols are linearly increasing under
some permutation. Formally, the LRC with noise parameter $\beta\in(0,\infty)$
takes as input $X^n\in\Bqty{0,1}^n$ and outputs $Y^n\in\Bqty{0,1}^n$ according
to the bitwise distribution
\begin{equation}\label{eq:lrc-p}
    p_{Y^n_i \gib X^n_i, T}(y\gib x,\tau) = \begin{cases}
        q_{\tau(i)} & y \neq x, \\
        1-q_{\tau(i)} & y = x,
    \end{cases}
\end{equation}
where $T\in S_n$ is the \emph{reliability ordering permutation} and the
associated bit-flip probabilities are, for $i\in[n]$,
\begin{equation}\label{eq:lrc-q}
    q_i = \frac{\upe^{-\beta i/n}}{1+\upe^{-\beta i/n}} \in (0,1/2).
\end{equation}
The permutation $T$ is sampled uniformly at random with each channel use. It is
straightforward to verify that the reliability of $Y^n_i$ is $\beta T(i)/n$.
In the context of the LRC, the difference between hard- and soft-decision
decoding amounts to whether or not the decoder is aware of the reliability
ordering permutation $T$. Given $T$, the decoder knows the (magnitudes) of the
LLR for each symbol. Without knowledge of $T$, the decoder only knows that $T$
is uniformly distributed.

Another way of distinguishing between the soft- and hard-decision settings is
to compare the effective distributions of the noise effect, i.e., the binary
sequence $X^n+Y^n$. Since each received symbol is equally unreliable to the
hard-decision decoder, all noise effects of the same Hamming weight are
equiprobable. Alternatively, because the soft-decision decoder knows the
reliability of each symbol, the probability of a given noise effect depends on
where the bit flips occur. In particular, it depends on a statistic called the
logistic weight \cite{DWM22}.

\begin{definition}\label{def:lw}
The \emph{logistic weight} of a sequence $x\in\Bqty{0,1}^n$ with respect to a
permutation $\tau\in S_n$ is
\begin{equation*}
    \wT(z) = \sum_{i:\,\tau(z)_i=1} i.
\end{equation*}
\end{definition}

We denote the soft-decision noise effect by $N^n$, i.e., the binary sequence
distributed according to posterior distribution of $X^n+Y^n$ given $T$, and the
hard-decision noise effect by $Z^n$, i.e., the sequence distributed according
to the corresponding prior distribution, assuming $T$ is uniformly distributed.
The following pair of propositions give the PMFs for these two distributions.
The PMF of $N^n$ is a function of the logistic weight with respect to $T$.

\begin{lemma}\label{lm:soft-pmf} In the LRC with parameter $\beta$, the
soft-decision noise effect $N^n$ has PMF \begin{equation*} p_{N^n}(x) = \frac{
\upe^{-\beta \wT(x) / n} } {\prod_{i=1}^n \pqty{1+\upe^{-\beta i/n}}},
\end{equation*}
where $\tau\in S_n$ is the realization of the reliability ordering permutation
for the given channel use.
\end{lemma}

\begin{proof}
By \cref{eq:lrc-p,eq:lrc-q},
\begin{equation*}
    p_{N^n}(x) = \prod_{i:\,\tau(x)_i = 1} q_i
        \prod_{i:\,\tau(x)_i = 0} (1-q_i).
\end{equation*}
Taking the logarithm,
\begin{align*}
    \ln p_{N^n}(x) &= \sum_{i:\,\tau(x)_i = 1} \ln q_i
        + \sum_{i:\,\tau(x)_i = 0} \ln(1-q_i) \\
    &= -\frac{\beta \wT(x)}{n} - \sum_{i=1}^n \ln(1+\upe^{-\beta i/n}).
\end{align*}
Exponentiating yields the desired expression.
\end{proof}

The PMF of $Z^n$ is given by averaging over all possible realizations of $T$.

\begin{lemma}\label{lm:hard-pmf}
In the LRC with parameter $\beta$, the hard-decision noise effect $Z^n$ has PMF
\begin{equation*}
    p_{Z^n}(x) = \frac{a^n_k(\beta)}{\binom{n}{k}
    \prod_{i=1}^n \pqty{1+\upe^{-\beta i/n}}},
\end{equation*}
where $k=\wH(x)$ and $a^n_k(\beta)$ denote the degree-$k$ elementary symmetric
polynomial in the $n$ variables $\upe^{-\beta i/n}$ for $i\in[n]$,
\begin{equation*}
    a^n_k(\beta) = \sum_{1\leq i_1<\dots<i_k\leq n}
    \upe^{-\beta(i_1+\dots+i_k)/n}.
\end{equation*}
\end{lemma}

\begin{proof}
Let $k=\wH(x)$. Since $T$ is uniformly distributed,
\begin{align}
\nonumber
    p_{Z^n}(x) &= \frac{1}{n!} \sum_{\tau\in S_n} p_{N^n\gib T}(x\gib \tau) \\
\label{eq:pmf-ank-sub}
    &= \frac{\sum_{\tau\in S_n} \upe^{-\beta \wT(x) /n}}
    {n! \prod_{i=1}^n \pqty{1+\upe^{-\beta i/n}}}.
\end{align}
As $\tau$ ranges over $S_n$, the sequence $\tau(x)$ takes on the value of each
Hamming weight $k$ sequence exactly $k!(n-k)!$ times. Letting $W_k$ denote the
set of Hamming weight $k$ sequences,
\begin{align}
\nonumber
    \sum_{\tau\in S_n} \upe^{-\beta \wT(x) /n} &= k!(n-k)! \sum_{w\in W_k}
    \upe^{-(\beta/n)\sum_{i=1}^n i w_i} \\
\label{eq:pmf-ank}
    &= k!(n-k)! a^n_k(\beta).
\end{align}
Noting that $n!=\binom{n}{k}k!(n-k)!$, substituting \cref{eq:pmf-ank} into
\cref{eq:pmf-ank-sub} yields the desired expression.
\end{proof}

\subsection{The LRC as an Approximation}\label{sec:approx}

\begin{figure*}
\centering
{\hfill%
\subfloat[normal distribution]
    {\includegraphics[width=\imgwidth]{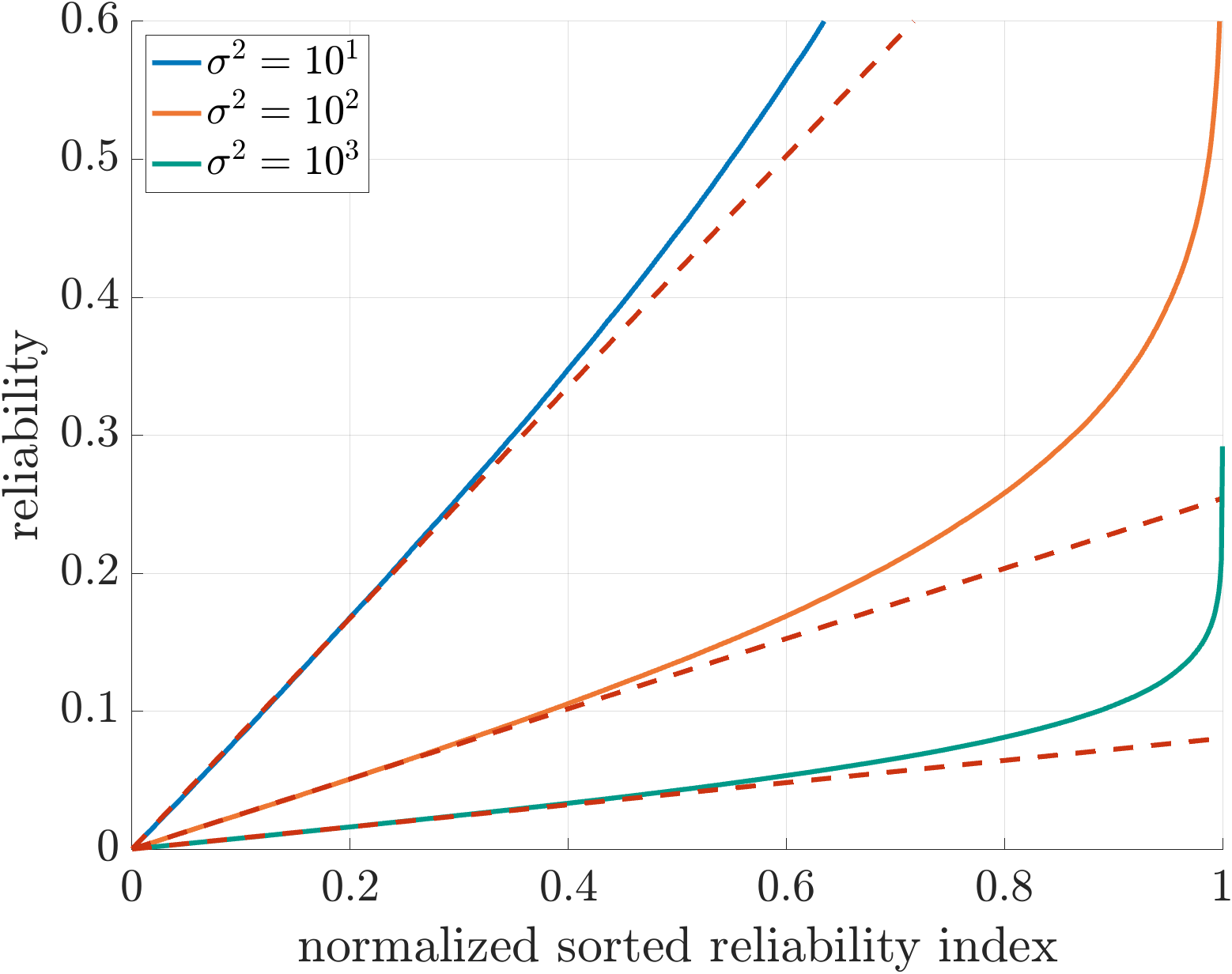}}%
\hfill%
\subfloat[logistic distribution]
    {\includegraphics[width=\imgwidth]{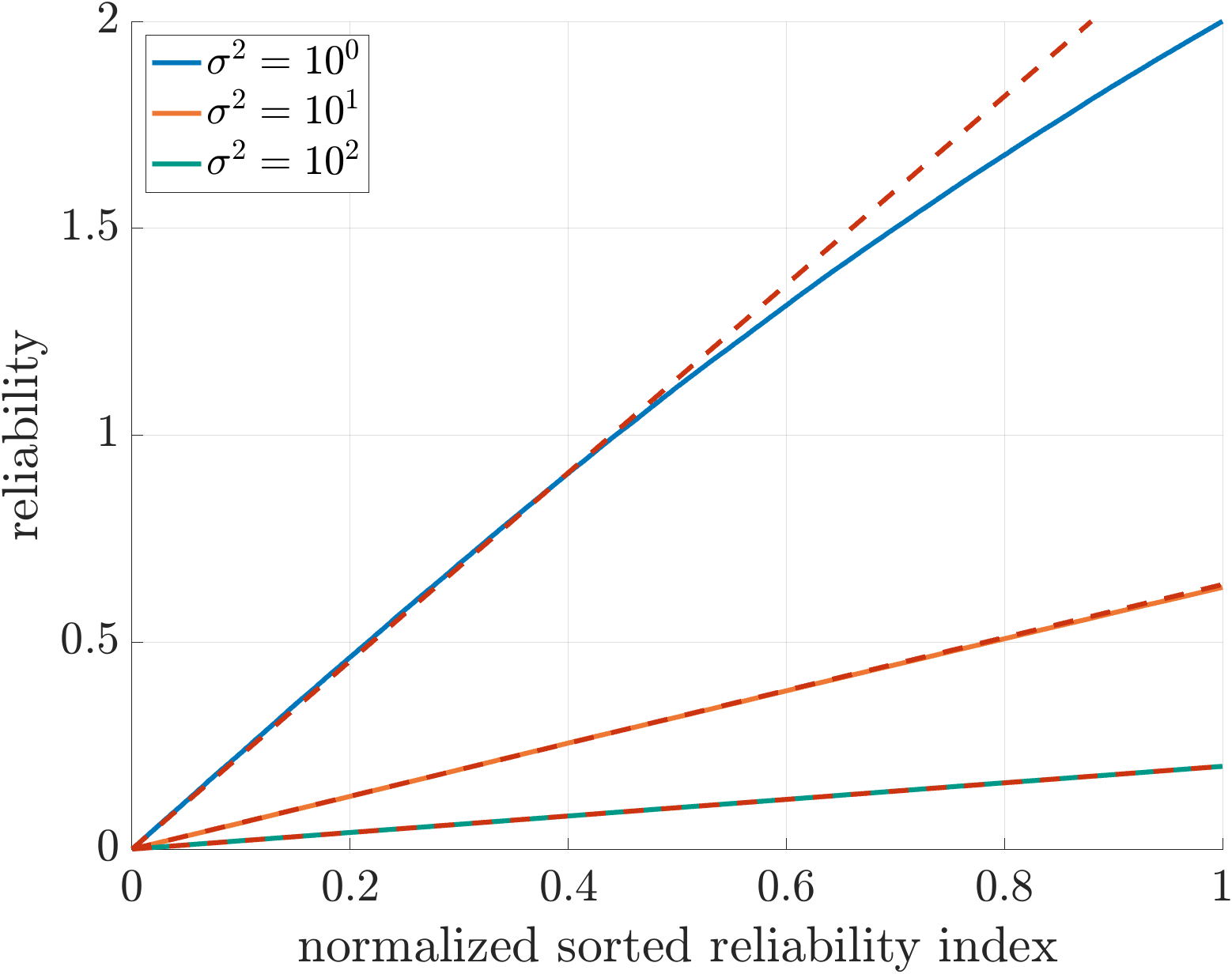}}%
\hfill}
\caption{Empirical sample of $n=2^{18}$ reliabilities, sorted in increasing
order, for noise following the normal and logistic distributions with mean 0
and variance $\sigma^2$. The horizontal axis is normalized by $n$. The red
dotted lines are hand-picked linear approximations showing that the sorted
reliabilities are initially approximately linear increasing. For both of these
noise distributions, the linear approximation is better over a wider range as
$\sigma^2$ increases.}
\label{fig:exp-llr-a}
\end{figure*}

\begin{figure*}
\centering
{\hfill%
\subfloat[laplace distribution (initial range)]
    {\includegraphics[width=\imgwidth]{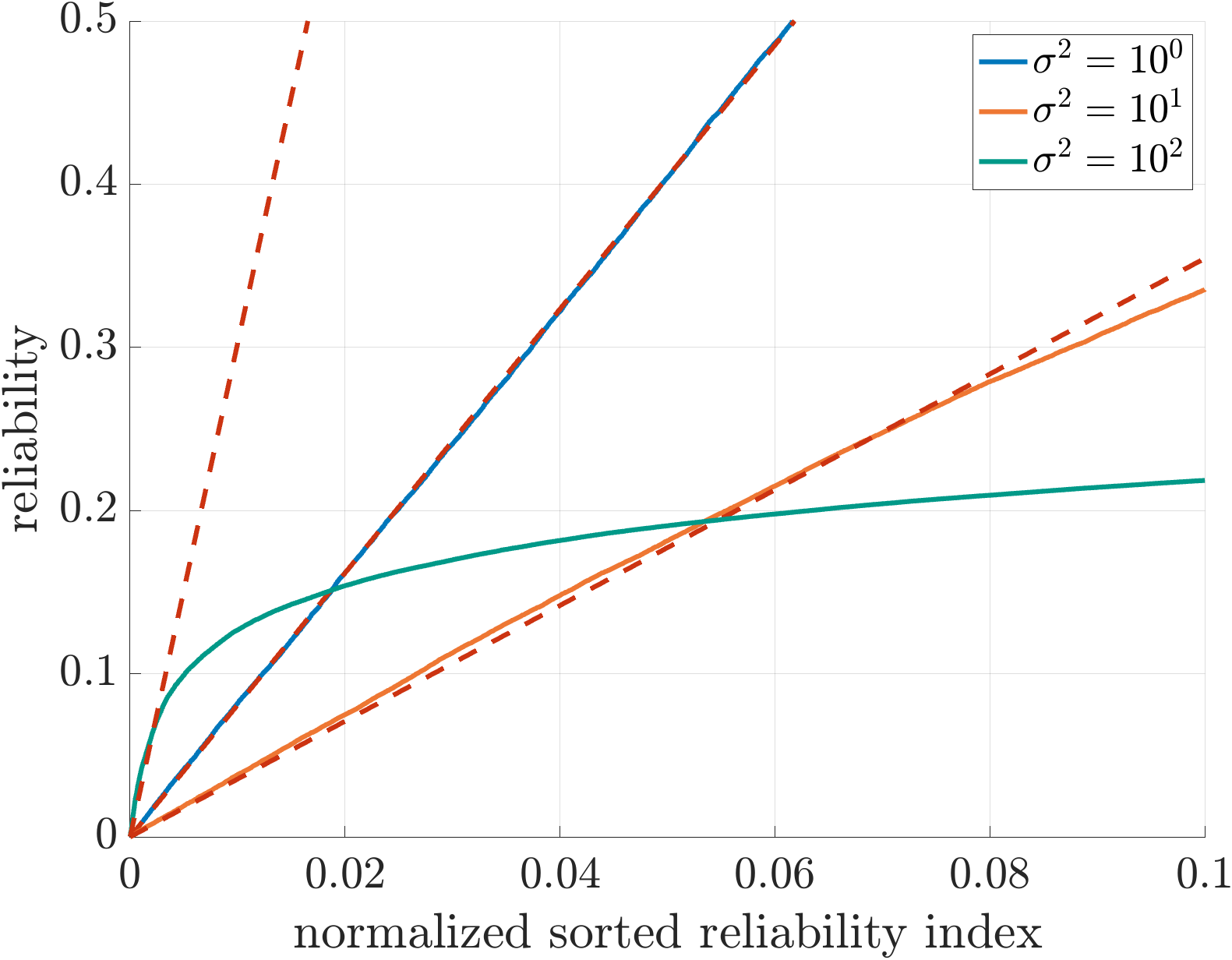}}%
\hfill%
\subfloat[Laplace distribution (full range)]
    {\includegraphics[width=\imgwidth]{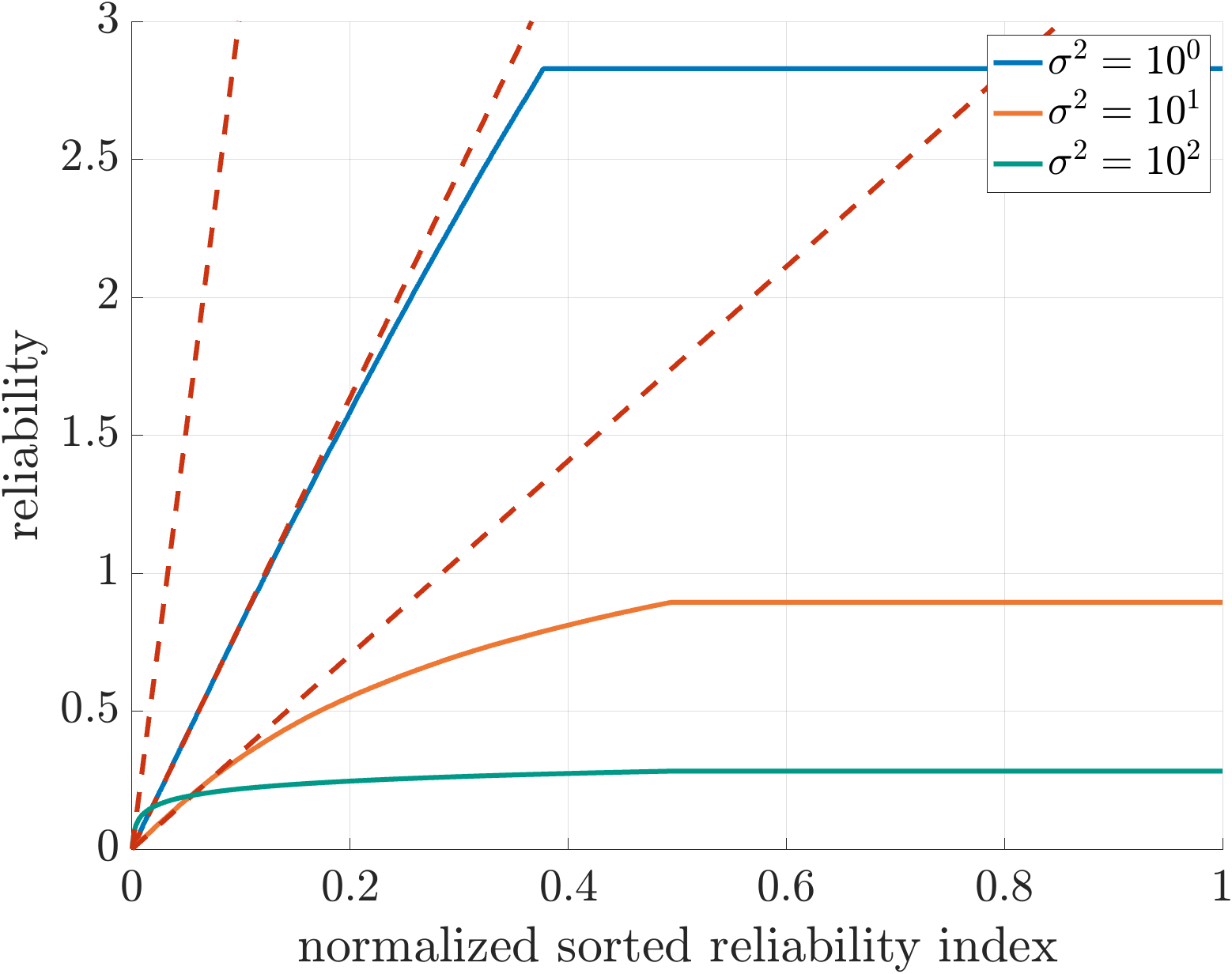}}%
\hfill}
\caption{Empirical sample of $n=2^{18}$ reliabilities, sorted in increasing
order, for noise following the Laplace distributions with mean 0 and variance
$\sigma^2$. The horizontal axis is normalized by $n$. The red dotted lines are
hand-picked linear approximations showing that the sorted reliabilities are
initially approximately linear increasing. The Laplace distribution is notable
for inducing reliabilities which are constant beyond given index after sorting.
Over an initial range prior to this transition, the reliabilities are
nonetheless approximately linearly increasing, although this range does shrink
as $\sigma^2$ grows.}
\label{fig:exp-llr-b}
\end{figure*}

The LRC captures the behavior of a wide range of channel noises and is usually
a better approximation as the noise variance increases. Informally, for a
general class of noise distributions, the reliabilities from a given block
transmission (of any length) are, over some initial range, roughly linearly
increasing when sorted in increasing order. This phenomenon is readily observed
in practice with multiple common noise distributions
(\cref{fig:exp-llr-a,fig:exp-llr-b}). The behavior of the most reliable symbols
depends more specifically on the particular noise distribution.

In this subsection we prove that, for a general class of ``location-scale''
noise distributions, the LLR PDF is asymptotically, as the noise variance
grows, linearly increasing in a neighborhood around zero. This implies that the
reliability CDF is also asymptotically linearly increasing over a (one-sided)
neighborhood around zero. It follows from a standard result on order statistics
that the sorted reliabilities must then be asymptotically, now as the block
length grows, given by the inverse reliability CDF \cite{DN04}. Thus, the
sorted reliabilities are asymptotically, in both block length and noise
variance, initially linearly increasing.


We consider an additive noise channel with binary input $X\in\Bqty{-1,+1}$,
continuous noise $N\in\R$, and continuous output $Y\in\R$ given by $Y=X+N$,
We assume that $X$ is uniformly distributed, such that
\begin{equation*}
    f_Y(y) = \frac{1}{2}\bqty{ f_N(y-1) + f_N(y+1) },
\end{equation*}
and that the noise satisfies the following assumptions.
\begin{assumption}\label{asm:noise}
The noise $N$ has a PDF of the form
\begin{equation*}
    f_N(x) = \frac{1}{\sigma}f_0\pqty{\frac{x}{\sigma}},
\end{equation*}
where $\sigma>0$ is the standard deviation of $N$ and $f_0$ is an even,
strictly log-concave, and $\ca{C}^4$ density.
\end{assumption}
Notable examples of distributions satisfying \cref{asm:noise} include the
normal distribution, the Laplace distribution, and the uniform distribution.
Note that evenness of $f_0$ implies that $N$ has mean zero.

Let $L\in\R$ be the LLR of the channel output $Y$ and let $\phi:\R\to\R$ be the
defined such that $L=\phi(Y)$, i.e., $\phi(y)$ is the LLR corresponding to the
output $y$. Under \cref{asm:noise},
\begin{equation*}
    \phi(y) = \ln f_0\pqty{\frac{y-1}{\sigma}} -
    \ln f_0\pqty{\frac{y+1}{\sigma}}.
\end{equation*}
Since $f_0$ is strictly log-concave, it follows that $\phi$ is monotonically
increasing and thus that the inverse $\phi^{-1}$ exists. Let $h:\R\to\R$ be
defined such that $f_L(l) = h(\phi^{-1}(l))$. In particular, $h(y) = f_Y(y) /
\phi'(y)$.

The following theorem formalizes the heuristic that the sorted reliabilities
are initially approximately linearly increasing under \cref{asm:noise}.

\begin{theorem}\label{thm:lrc-approx}
If $N$ satisfies \cref{asm:noise}, then, as $\sigma\to\infty$ and
for $\epsilon=\ca{O}(\sigma^{-3/2})$,
\begin{equation*}
    \sup_{\abs{l}\leq \epsilon} \abs{ f_L(l) - f_L(0) } = \ca{O}(1).
\end{equation*}
\end{theorem}

The fact that \cref{thm:lrc-approx} requires that $\epsilon$ goes to 0 does not
imply that the sorted reliabilities are only linearly increasing over a range
which is negligible for large $\sigma$. Intuitively, the typical reliability is
also decreasing as $\sigma$ grows, as is readily observable in
\cref{fig:exp-llr-a,fig:exp-llr-b}. This implies that the range over which the
reliability CDF is meaningfully less than 1 is also decreasing, and it is
precisely this regime which is treated by \cref{thm:lrc-approx}. However, the
CDF also becomes more linear over that regime as it shrinks, and the LRC
generally becomes a better approximation as a result. The qualitative
difference in behavior for the Laplace distribution, illustrated in
\cref{fig:exp-llr-b}, is due to the fact that the reliabilities induced by that
noise distribution are constant after a given point. Since the transition point
is decreasing with $\sigma$, this is one example of a noise distribution for
which the LRC is a worse approximation as as $\sigma$ increases.

The key ingredient in the proof of \cref{thm:lrc-approx} is the following
lemma describing the concavity of $f_L$ at $0$.

\begin{lemma}\label{lm:o3}
If $N$ satisfies \cref{asm:noise}, then $f''_L(0) = \ca{O}\pqty{\sigma^{3}}$ as
$\sigma\to\infty$.
\end{lemma}

\begin{proof}
For conciseness, we abuse notation and write $y(l) = \phi^{-1}(l)$. We have
that $y'(l) = 1/\phi'(y(l))$, and the first and second derivatives of $f_L(l) =
h(y(l))$ are
\begin{align*}
    f'_L(l) &= \frac{ h'(y(l)) }{ \phi'(y(l)) }, \\
    f''_L(l) &= \frac{ h''(y(l)) }{ \bqty{\phi'(y(l))}^2 } -
        \frac{ h'(y(l)) \phi''(y(l))}{ \bqty{\phi'(y(l))}^3 }.
\end{align*}

Since $\phi$, and hence $y$, are odd, $y(0)=0$. Since $h$ is even,
$h'(y(0))=h'(0)=0$. Since $\phi$ is monotonically increasing, $\phi'(0)>0$.
Together, these imply that
\begin{equation}\label{eq:f2l-raw}
    f''_L(0) = \frac{ h''(0) }{ \bqty{\phi'(0)}^2 }
\end{equation}
The second derivative of $h(y)=f_Y(y)/\phi'(y)$ is
\begin{equation*}
\begin{split}
    h''(y) &= \frac{ f''_Y(y) }{ \phi'(y) }
        - \frac{ 2 f'_Y(y) \phi''(y) }{ \bqty{\phi'(y)}^2 } \\
        &\quad - \frac{ f_Y(y)\phi'''(y) }{ \bqty{\phi'(y)}^2 }
        + \frac{2 f_Y(y)\bqty{\phi''(y)}^2 }{ \bqty{\phi'(y)}^3 }.
\end{split}
\end{equation*}
Since $f_Y$ is even and $\phi$ is odd,
\begin{equation*}
    h''(0) = \frac{ f''_Y(0) }{ \phi'(0) }
        - \frac{ f_Y(0) \phi'''(0) }{ \bqty{\phi'(0)}^2 }.
\end{equation*}
Substituting into \cref{eq:f2l-raw}
\begin{equation}\label{eq:f2l-clean}
    f''_L(0) = \frac{ f''_Y(0)\phi'(0) - f_Y(0)\phi'''(0) }
        { \bqty{\phi'(0)}^4 }.
\end{equation}
Noting that $f_0$ is even, the quantities appearing in \cref{eq:f2l-clean} are
\begin{align*}
    f_Y(0) &= \sigma^{-1}f_0\pqty{\sigma^{-1}}, \\
    f''_Y(0) &= \sigma^{-3}f''_0\pqty{\sigma^{-1}}, \\
    \phi'(0) &= -2\sigma^{-1} \bqty{
        \frac{ f'_0\pqty{\sigma^{-1}} } { f_0\pqty{\sigma^{-1}} }
    }, \\
\begin{split}
    \phi'''(0) &= -2\sigma^{-3}\left[
        \frac{ f'''_0\pqty{\sigma^{-1}} }{ f_0\pqty{\sigma^{-1}} }
        - \frac{ 3 f''_0\pqty{\sigma^{-1}} f'_0\pqty{\sigma^{-1}} }
            { \bqty{f_0\pqty{\sigma^{-1}}}^2 } \right. \\
        &\qquad\qquad\qquad \left.+ \frac{ 2\bqty{f'_0\pqty{\sigma^{-1}}}^3 }
            { \bqty{f_0\pqty{\sigma^{-1}}}^3 } \right].
\end{split}
\end{align*}
The Taylor expansions of $f_0$ and its derivatives at $y=0$, evaluated at
$\sigma^{-1}$, are
\begin{align*}
    f_0\pqty{\sigma^{-1}} &= f_0(0) + \frac{1}{2}\sigma^{-2}f''_0(0) +
        \ca{O}\pqty{\sigma^{-4}}, \\
    f'_0\pqty{\sigma^{-1}} &= \sigma^{-1}f''_0(0) +
        \ca{O}\pqty{\sigma^{-3}}, \\
    f''_0\pqty{\sigma^{-1}} &= f''_0(0) + \ca{O}\pqty{\sigma^{-2}}, \\
    f'''_0\pqty{\sigma^{-1}} &= \sigma^{-1}f^{(4)}(0) +
        \ca{O}\pqty{\sigma^{-3}}.
\end{align*}
Correspondingly,
\begin{align*}
    f_Y(0) &= \sigma^{-1}f_0(0) + \ca{O}\pqty{\sigma^{-3}}, \\
    f''_Y(0) &= \sigma^{-3}f''_0(0) + \ca{O}\pqty{\sigma^{-5}}, \\
    \phi'(0) &= -2\sigma^{-2}\bqty{\frac{f''_0(0)}{f_0(0)}} +
        \ca{O}\pqty{\sigma^{-4}}, \\
    \phi'''(0) &= \ca{O}\pqty{\sigma^{-4}}.
\end{align*}
Substituting into \cref{eq:f2l-clean}, the numerator is
$\ca{O}\pqty{\sigma^{-5}}$ and the denominator is $\ca{O}\pqty{\sigma^{-8}}$.
Thus, $f''_L(0) = \ca{O}\pqty{\sigma^{3}}$.
\end{proof}

\begin{proof}[of \cref{thm:lrc-approx}]
The evenness of $f_L$ implies that the second-order Taylor expansion around
$l=0$ is simply $f_L(l) = f_L(0) + \ca{O}(l^2)$, with the approximation error
over the interval $[-\epsilon,\epsilon]$ bounded by
\begin{equation*}
    \sup_{\abs{l}\leq\epsilon} \abs{f_L(l) - f_L(0)} \leq \frac{1}{2}
    \sup_{\abs{l}\leq \epsilon} \abs{f''_L(l)} \epsilon^2.
\end{equation*}
Since $f''_L(l)$ is continuous at $l=0$ and $\epsilon\to 0$ as
$\sigma\to\infty$, we have, for $\sigma$ sufficiently large,
\begin{equation*}
    \sup_{\abs{l}\leq \epsilon} \abs{f''_L(l)} \leq 2\abs{f''_L(0)}.
\end{equation*}
By \cref{lm:o3}, the magnitude of $f''_L(0)$
is $\ca{O}\pqty{\sigma^{3}}$ as $\sigma\to\infty$, and hence
\begin{equation*}
    \sup_{\abs{l}\leq\epsilon} \abs{f_L(l) - f_L(0)} \leq
    \ca{O}\pqty{\sigma^3\epsilon^2} = \ca{O}(1).
\end{equation*}
\end{proof}

\subsection{The Logistic Weight}\label{sec:weight}

\begin{figure}
\centering
\subfloat[The logistic coefficient
$a(n,w)$.]{\includegraphics[width=\linewidth]{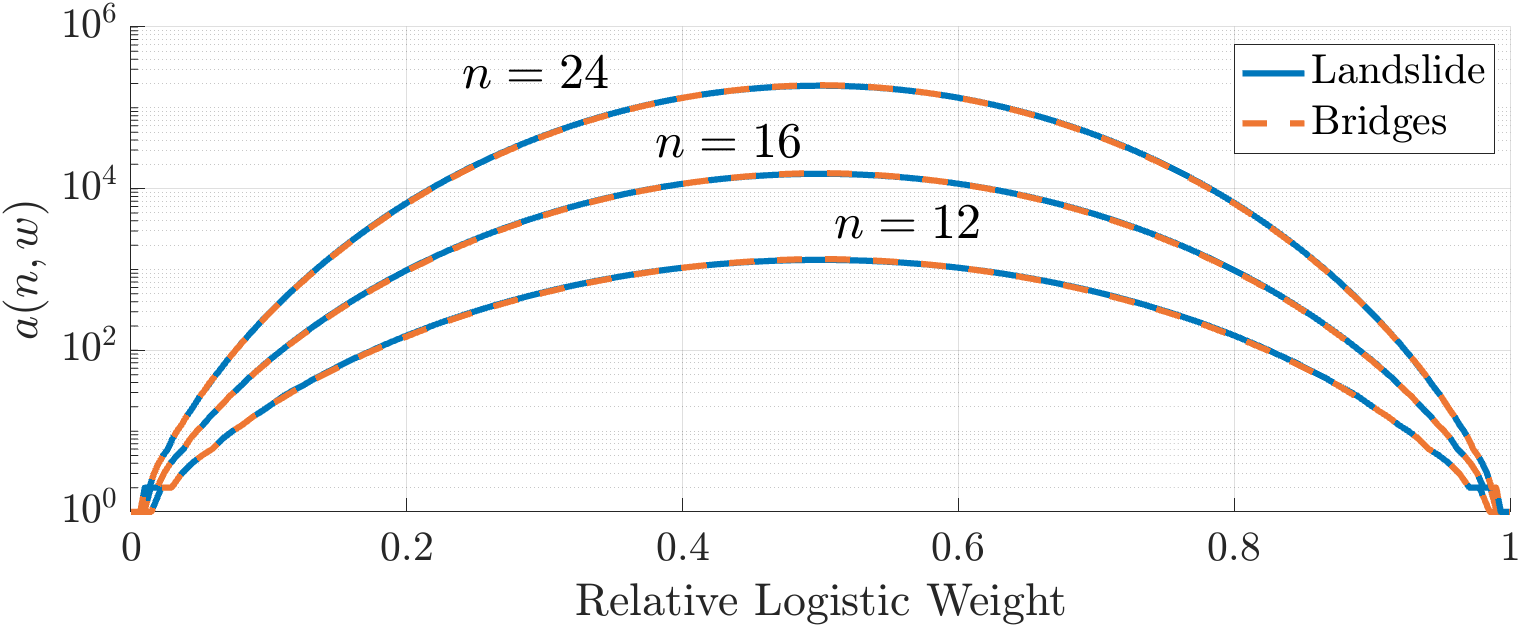}}\\
\subfloat[Bridges' approximation
error]{\includegraphics[width=\linewidth]{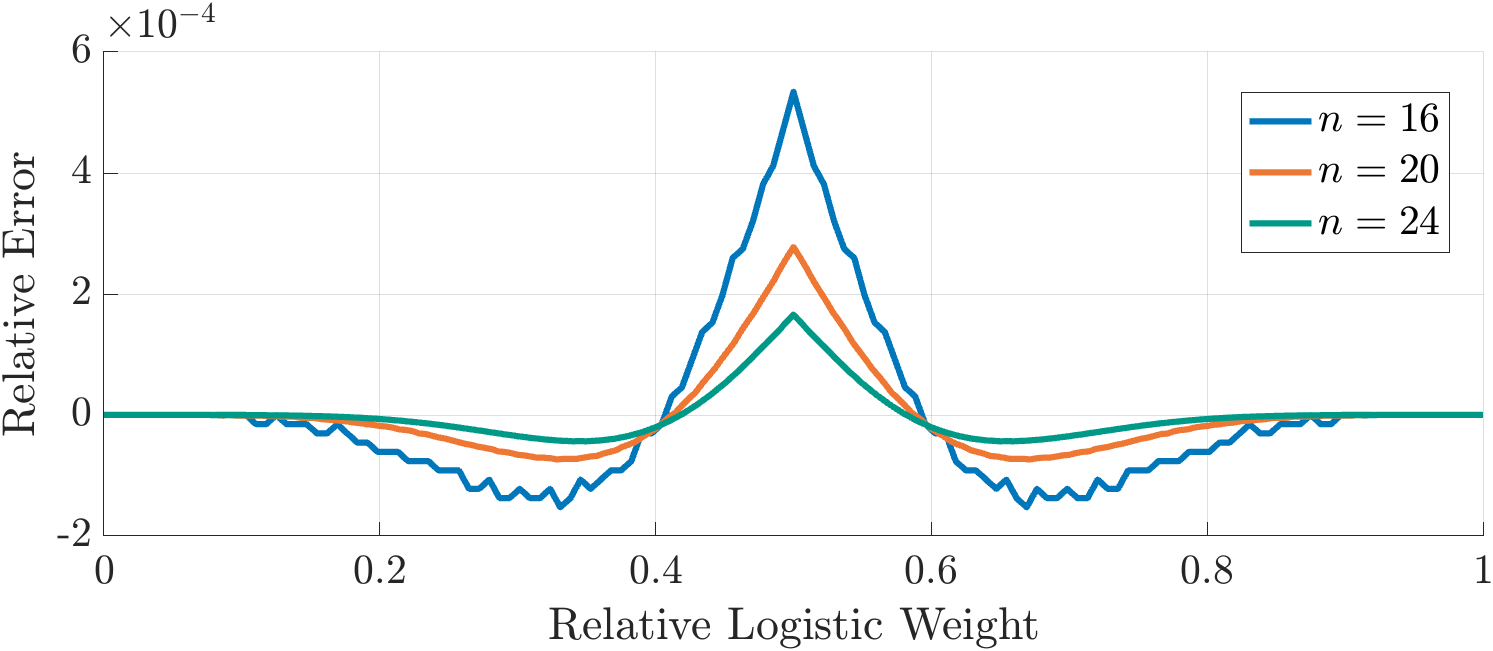}}\\
\caption{Comparison between Bridges' approximation (\cref{thm:bridges}) and the
exact value of $a(n,w)$, computed with the Landslide algorithm, for small $n$.
The horizontal axes are normalized by $n(n+1)/2$, the maximum weight. The
vertical axis of (b) is normalized by $2^n$, the total number of sequences.}
\label{fig:bridges}
\end{figure}

The logistic weight is intimately connected to the LRC in much the same way
that the Hamming weight is connected to the BSC. The Hamming weight of a length
$n$ sequence defines its type in the context of the BSC, and a core feature for
decoding is that there are $n+1$ such types. In contrast, the logistic weight
of a sequence, which defines its type in the context of soft-decision decoding
in the LRC, has $n(n+1)/2$ such types. This finer partition of the space of
sequences directly corresponds to the greater resolution provided by the soft
information. We establish here some essential properties of the logistic weight
and the enumeration of sequences of each type, emphasizing the parallels to
Hamming weight and the BSC throughout. In the context of the LRC, the logistic
weight is usually taken with respect to the reliability ordering permutation.
In some cases, however, the specific permutation does not matter, e.g., when
counting the number of sequences of a given length and logistic weight.

Because noise effects with the same logistic weight are equiprobable, the
logistic weight plays a similar role in the LRC as the Hamming weight does in
the BSC. A consequence of \cref{lm:soft-pmf} is the following identity: for all
$\beta>0$,
\begin{equation}\label{eq:logistic-coeff}
    \sum_{w=0}^{n(n+1)/2} \frac{ a(n,w) \upe^{-\beta \wT(z) / n} }
        {\prod_{i=1}^n \pqty{1+\upe^{-\beta i/n}}} = 1,
\end{equation}
where $a(n,w)$ is the number of length $n$ sequences with logistic weight $w$.
We refer to $a(n,w)$ as the logistic coefficient.
\Cref{eq:logistic-coeff} parallels the familiar identity which arises from the
BSC and the binomial distribution: for all $p\in[0,1]$,
\begin{equation*}
    \sum_{w=0}^n \binom{n}{w}p^w(1-p)^{n-w}=1.
\end{equation*}

In addition to their probabilistic interpretations as normalizing constants for
specific distributions, both the binomial and logistic coefficients have
combinatorial interpretations. The binomial coefficient $\binom{n}{w}$ counts
the subsets of $w$ elements of a set of $n$ elements. The logistic
coefficient is related to integer partitions: $a(n,w)$ is equal to the number
of partitions of $w$ with distinct parts and largest part at most $n$.

The logistic coefficient, like the binomial coefficient, is symmetric in $w$.
\begin{proposition}
For all $n\geq 1$ and $0\leq w\leq n(n+1)/2$,
\begin{equation*}
    a(n,w) = a\pqty{n, \frac{n(n+1)}{2}-w}.
\end{equation*}
\end{proposition}

\begin{proof}
For $r\in\Bqty{0,1}$ and $0\leq w \leq n(n+1)/2$, the set
\begin{align*}
    \ca{Z}_r(w) &= \Bqty{z\in\Bqty{0,1}^n: \sum_{i:\,z_i=r} i=w},
\end{align*}
i.e., with respect to the identity permutation, $\ca{Z}_1(w)$ is set of length
$n$ sequences with logistic weight $w$, while $\ca{Z}_0(w)$ is the set of
those with logistic weight $n(n+1)/2 - w$. By symmetry, $\abs{\ca{Z}_1(w)}
= \abs{\ca{Z}_0(w)}$ for all $w$.
\end{proof}

Although the logistic coefficient may not be expressible algebraically, the
combinatorial interpretation yields methods for both computing and
approximating $a(n,w)$. The landslide algorithm \cite{DWM22} enumerates all
length $n$ sequences of logistic weight $w$. For values of $n$ and $w$ for
which this algorithm is impractical, the following asymptotic approximation,
which is a reparameterization of a result due to Bridges \cite{Bri20}, is
extremely accurate (\cref{fig:bridges}).

\begin{theorem}[\cite{Bri20}]\label{thm:bridges}
Define $\beta: \pqty{\sqrt{2},\infty} \to \pqty{-\infty,\frac{\pi}{2\sqrt{3}}}$
as an implicit function of $t$ such that
\begin{equation*}
    1 = \int_0^t \frac{u\upe^{-\beta u}}{1+\upe^{-\beta u}} \dd{u}.
\end{equation*}
Let
\begin{align*}
    A(t) &= \frac{ \upe^{\frac{\beta t}{2}} + \upe^{\frac{-\beta t}{2}} }
    { 2 } \sqrt{\frac{\beta'(t)}{\pi t}}, \\
    B(t) &= 2\beta + t\ln\pqty{1+\upe^{-\beta t}}.
\end{align*}
Then,
\begin{equation}\label{eq:bridges-est}
    a(n,w) \sim \frac{ A \pqty{\frac{n}{\sqrt{w}}} } {w^{3/4}}
        \upe^{ B\pqty{ \frac{n}{\sqrt{w}} } \sqrt{w} }.
\end{equation}
\end{theorem}

When $w$ is near $n(n+1)/2$, it is possible that $\frac{n}{\sqrt{w}} <
\sqrt{2}$, for which Bridges' $\beta$ function is not defined. Nonetheless,
since $a(n,w)$ is symmetric in $w$ and, for all $n\geq 1$,
\begin{equation*}
    \frac{n}{\sqrt{\frac{n(n+1)}{4}}} > \sqrt{2},
\end{equation*}
Bridges' approximation can be used for all values of $w$.

\section{Maximum Likelihood Decoders}\label{sec:ml}

There exist explicit hard- and soft-decision ML decoding algorithms for the
LRC. These decoders are readily described in the framework of GRAND
\cite{DLM19}, a family of code-agnostic channel decoding algorithms. We give
here a brief overview of the principles which are sufficient for a complete
formal description of both the soft-decision (\cref{thm:ml-soft}) and
hard-decision ML decoders (\cref{thm:ml-hard}) for the LRC, as well as for a
detailed analysis of the probability of a decoding error in the sequel.

In any additive noise channel, identifying the code word which maximizes the
likelihood of the received transmission is equivalent to identifying the noise
effect which maximizes that same likelihood. Formally, denoting the code by
$\ca{C}\subset\Bqty{0,1}^n$,
\begin{align*}
    c^{\ast} &= \argmax \Bqty{p_{Y^n\gib X^n}(y^n\gib c) : c\in\ca{C}} \\
    &= \argmax \Bqty{p_{N^n}(Y^n - c) : c\in\ca{C}}.
\end{align*}
Given a statistical model for the noise (which, for the purposes of specifying
an algorithm, need not correspond to the true channel noise distribution), all
possible noise effects can be rank ordered by probability. The first noise
effect in this order which yields a code word when subtracted from the received
sequence is the most likely noise effect under the given model, and the
corresponding code word is the most likely decoding.

The invertible map $G:\Bqty{0,1}^n\to[2^n]$ which rank orders noise effects is
referred to as a guessing function, and the behavior of GRAND can be analyzed
in the information theoretic context of guesswork \cite{Mas94,Ari96,MS04}. When
the guessing function is optimal, i.e., the statistical model does correspond
to the true channel noise distribution and noise effects are guessed in
non-increasing order of probability, then it is an ML decoder. Thus, an ML
decoder for a particular channel can be completely specified by an optimal
guessing function for its noise effect distribution. For hard-decision
decoding, the guessing function must be optimal with respect to the prior noise
effect distribution. For soft-decision decoding, it must be optimal with
respect to the posterior distribution given the received transmission and the
corresponding soft information.

The soft-decision ML decoder for the LRC guesses noise effects in order of
increasing logistic weight with respect to the reliability ordering
permutation. This algorithm is ORBGRAND \cite{DWM22}, originally proposed as an
approximate soft-decision ML decoder and later shown to be almost
capacity-achieving for the real-valued AWGN channel \cite{Liu+22}.

\begin{theorem}\label{thm:ml-soft}
For any $\tau\in S_n$, let $\GT:\Bqty{0,1}^n\to[2^n]$ be a guessing
function such that for all $x_1,x_2\in\Bqty{0,1}^n$,
\begin{align*}
    \GT(x_1) < \GT(x_2) &\implies \wT(x_1) \leq \wT(x_2), \\
    \wT(x_1) < \wT(x_2) &\implies \GT(x_1) < \GT(x_2).
\end{align*}
Then, the GRAND algorithm using $\GT$ as a guessing function is an
soft-decision ML decoder for the LRC given that the reliability ordering
permutation is $\tau$.
\end{theorem}

\begin{proof}
We show that $\GT$ is an optimal guessing function for $N^n$.
\Cref{lm:soft-pmf} implies that
\begin{equation*}
    p_{N^n}(x) \propto \upe^{-\beta \wT(x) / n}.
\end{equation*}
Since this function is strictly decreasing in $\wT(x)$,
\begin{align*}
    p_{N^n}(x_1) > p_{N^n}(x_2) &\implies \wT(x_1) < \wT(x_2) \\
    &\implies \GT(x_1) < \GT(x_2).
\end{align*}
Similarly,
\begin{align*}
    \GT(x_1) < \GT(x_2) &\implies \wT(x_1) \leq \wT(x_2) \\
    &\implies p_{N^n}(x_1) \geq p_{N^n}(x_2).
\end{align*}
\end{proof}

The hard-decision ML decoder for the LRC guesses noise effects in order of
increasing Hamming weight. This corresponds to the original version of GRAND
\cite{DLM19}, first proposed as a general hard-decision ML decoder which, for
the noise distribution of the BSC, guesses by Hamming weight. Note, however,
that the hard-decision LRC is not equivalent to a BSC: although the marginal
distribution of the noise effect is identical for each bit, the bits are not
independent. Nonetheless, because the two channels do have the same optimal
guessing function, this does imply that any hard-decision ML decoder for the
BSC is also a hard-decision ML decoder for the LRC. By considering the
operation of GRAND algorithms specifically, however, both hard- and
soft-decision decoding in the LRC can be tackled with a common set of
techniques.

\begin{theorem}\label{thm:ml-hard}
Let $\GH:\Bqty{0,1}^n\to[2^n]$ be a guessing function such that, for all
$x_1,x_2\in\Bqty{0,1}^n$,
\begin{align*}
    \GH(x_1) < \GH(x_2) &\implies \wH(x_1) \leq \wH(x_2), \\
    \wH(x_1) < \wH(x_2) &\implies \GH(x_1) < \GH(x_2).
\end{align*}
Then, the GRAND algorithm using $\GH$ as a guessing function is a hard-decision
ML decoder for the LRC.
\end{theorem}

\begin{proof}
We show that $\GH$ is an optimal guessing function for $Z^n$. By
\cref{lm:hard-pmf},
\begin{equation*}
    p_{Z^n}(x) \propto \frac{a^n_k(\beta)}{\binom{n}k} = E_k,
\end{equation*}
where $k=\wH(x)$. To show that $p_{Z^n}(x)$ is strictly decreasing in $\wH(x)$,
it suffices to show that $E_k$ is strictly decreasing in $k$. Since the factors
in $a^n_k(\beta)$ are all distinct, Maclaurin's inequality \cite{HLP52} yields
\begin{equation*}
    E_k > \pqty{E_{k+1}}^{k/(k+1)}.
\end{equation*}
Since $E_k\in(0,1)$ for all $k$,
\begin{equation*}
    \pqty{E_{k+1}}^{k/(k+1)} > E_{k+1},
\end{equation*}
and hence $E_k > E_{k+1}$. The remainder of the proof follows the same logic as
that of \cref{thm:ml-soft}.
\end{proof}

In the sequel, we denote by $\GT$ the optimal guessing function for
soft-decision decoding, with the understanding the $\tau$ refers to the
realization of the reliability ordering permutation. We continue to denote by
$\GH$ the optimal guessing function for hard-decision decoding, and we simply
use $G$ to refer to any other generic guessing function.

Note that neither the soft- nor hard-decision ML decoder depends on
$\beta$. Their performance will depend on the noise level, as we will show, but
not their optimality.

\section{Large Deviation Principles for\texorpdfstring{\\}{ }Guesswork in the
LRC}\label{sec:ldp}

A key benefit of the fact that the ML decoders for the LRC are expressible as
GRAND algorithms with explicit guessing functions is that the error behavior is
describable in the mathematical language of large deviations. In this section,
we leverage both standard large deviations techniques and GRAND-specific
results to establish large deviation principles (LDP) for the number of guesses
made by the hard- and soft-decision ML decoders. In \cref{sec:error}, these
LDPs are used to derive both error exponents (for the probability of
incorrectly decoding below capacity) and success exponents (for the probability
of correctly decoding above capacity) for these decoders. Proofs of error
exponents have more traditionally been handled using techniques based on the
method of types and the notion of typical sets \cite{Csi98}. One notable
benefit of the alternative large deviations approach that we take here is that
error and success exponents are captured in a single coherent framework. We
begin by giving a brief, informal overview of the theory of large deviations,
with the goal of imparting an intuitive understanding of our results. For a
more thorough but still relatively informal introduction, see \cite{Tou09}, and
for a complete formal treatment, see \cite{DZ10,Var84,DS01}.

At a high level, the theory of large deviations considers the probability that
the realization of a random variable in a sequence is far from its expectation,
i.e., the probability of observing a large deviation. The perspective taken is
inherently asymptotic. We consider a infinite sequence of random variables
$A^n$, indexed by $n\in\N$. We refer to such a sequence as a \emph{process}.
For our purposes, we may simply let $A^n$ be real-valued. Informally, such a
sequence satisfies an LDP with rate function $I_A$ if, as $n\to\infty$,
\begin{equation*}
    \Pr{A^n \in (a,b)} \approx \exp(-n \inf_{x\in (a,b)} I_A(x)).
\end{equation*}
Loosely, the rate function quantifies the exponential rate at which the
probability of $A_n$ taking values over any interval is decaying asymptotically
with $n$. In general, there exists some point $x^\ast$ for which
$I_A(x^\ast)=0$, which implies that the probability that $A^n\approx x^\ast$ is
not decaying as $n$ grows. This asymptotic concentration is expressed by
classical results such as the central limit theorem. The theory of large
deviations generalizes such results by quantifying the decay rate of the
probability of any given atypical observation.

Let $N^n$ denote the noise effect\footnote{Whether we consider a soft-decision
or hard-decision noise effect is not relevant to this discussion. We use the
notation $N^n$, elsewhere used to denote a soft-decision noise effect,
arbitrarily.} and let $G$ be the optimal guessing function for $N^n$. We refer
to $G(N^n)\in[2^n]$, the position of the noise effect in the rank ordering
induced by $G$, as the \emph{guesswork} of $N^n$. Let $U^n$ denote the first
sequence guessed by $G$ corresponding to an \emph{incorrect} code word, i.e.
$Y^n+U^n$ is a code word but $U^n\neq N^n$. GRAND produces the correct decoding
if and only if $G(N^n) < G(U^n)$. Thus, the asymptotic probability of an ML
decoding error is determined by distribution of $G(N^n)$ and $G(U^n)$ in the
large block length limit.

The optimal guesswork process $\Bqty{n^{-1}\ln G(N^n)}$ has been shown to
satisfy an LDP for a general class of noise distributions \cite{CD12}. The fact
that we consider the exponent of the guesswork rather than the guesswork
directly is effectively due to the fact that the total number of sequences is
growing exponentially in $n$. To establish that $n^{-1}\ln G(N^n)$ satisfies an
LDP, it suffices to show that its scaled cumulant generating function (sCGF) is
expressible as a particular function of the R\'enyi entropy rate of $N^n$.

The sCGF $\Lambda_A$ of a general, real-valued random process $A^n$ is defined
to be \begin{equation*} \Lambda_A(\alpha) = \lim_{n\to\infty} \frac{1}{n}
\ln\Exp{\upe^{\alpha n
    A^n}}.
\end{equation*}
When the sCGF exists and satisfies some regularity conditions, $A^n$ satisfies
an LDP with a rate function $I_A$ given by the Legendre-Fenchel transform of
the sCGF,
\begin{equation*}
    I_A(x) = \sup_{\alpha\in\R}\Bqty{x\alpha - \Lambda_A(\alpha)}.
\end{equation*}

The choice of working with natural logarithms is largely conventional. We carry
out most of our analysis with natural logarithms for convenience, but
ultimately the error exponents and rate functions for guesswork processes are
more readily interpretable when expressed in bits. For the transformations
between nats and bits, see \cref{eq:b2-scgf,eq:b2-rate}.

For a random sequence $X^n$ of letters drawn from a finite alphabet, the
R\'enyi entropy of order $\alpha\in(0,1)\cup(1,\infty)$ is defined to be (in
nats)
\begin{equation*}
    H_{\alpha}(X^n) = \frac{1}{1-\alpha}\ln\pqty{\sum_x \Pr{X^n = x}^\alpha},
\end{equation*}
and the R\'enyi entropy rate of order $\alpha$ is given by the limit
\begin{equation*}
    H_{\alpha}(X) = \lim_{n\to\infty} \frac{1}{n} H_{\alpha}(X^n).
\end{equation*}
The R\'enyi entropy rate is generalization of, among other quantities, the
min-entropy rate $H_{\up{min}}(X)$ and the Shannon entropy rate $H_1(X)$. In
particular,
\begin{align*}
    H_{\up{min}}(X) &= \lim_{\alpha\to\infty} H_{\alpha}(X), \\
    H_1(X) &= \lim_{\alpha\to 1} H_{\alpha}(X).
\end{align*}
We also denote by $h:[0,1]\to[0,\ln 2]$ the usual binary entropy function (in
nats),
\begin{equation*}
    h(p) = -p \ln p - (1-p)\ln(1-p).
\end{equation*}
We slightly abuse notation by not distinguishing whether the various entropies
are in bits or in nats. The choice of logarithm will be clear from context.

Finally, the following integral appears repeatedly throughout our
analysis, and so we denote it by the following function of
$r\in[0,1]$ and $\gamma>0$,
\begin{equation*}
    J(r;\gamma) = \int_0^1 \ln(1+r\upe^{-\gamma x}) \dd{x}.
\end{equation*}

\begin{figure*}
\centering
{\hfill%
\subfloat[sCGFs]
    {\includegraphics[width=\imgwidth]{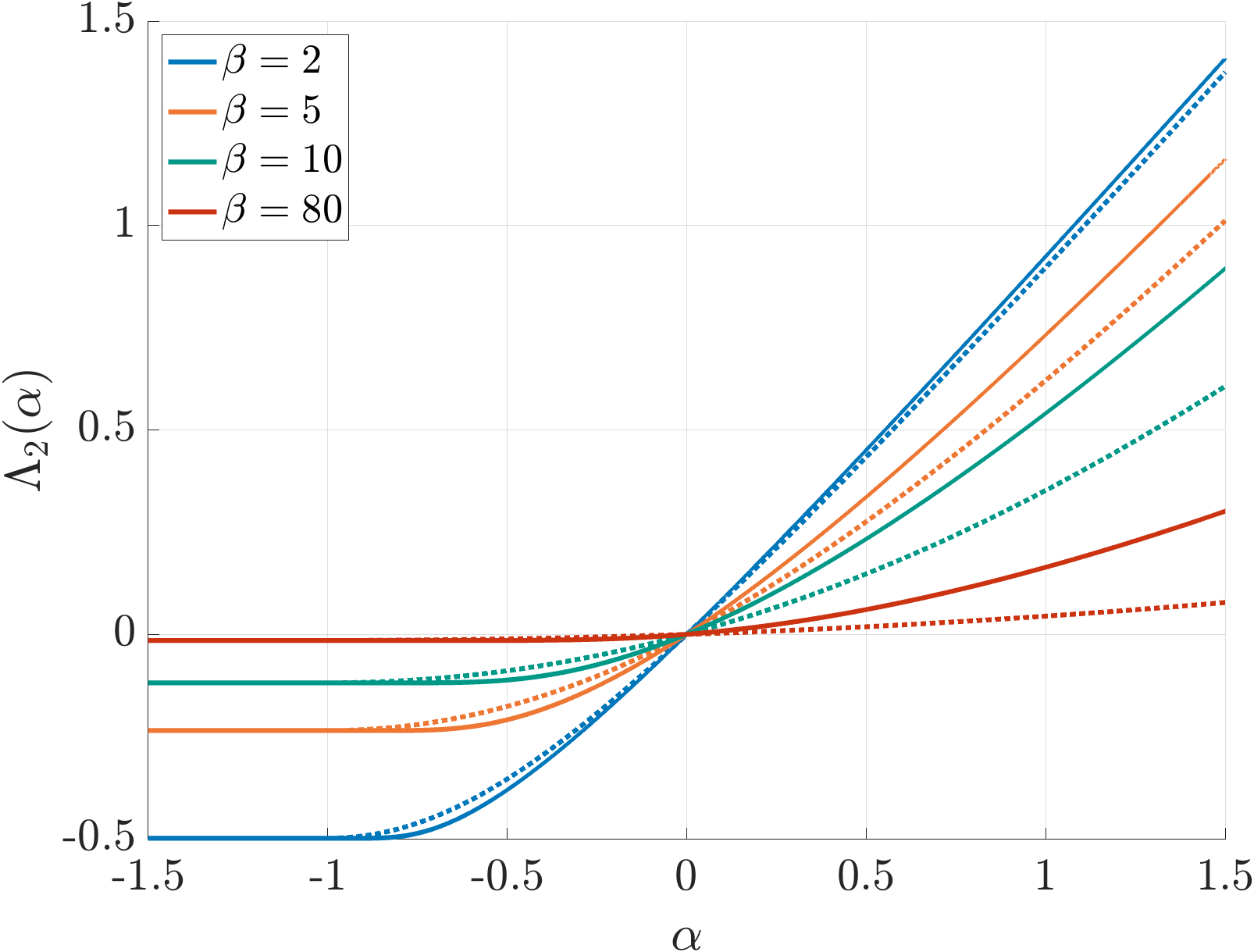}\label{fig:scgf}}%
\hfill%
\subfloat[rate functions]
    {\includegraphics[width=\imgwidth]{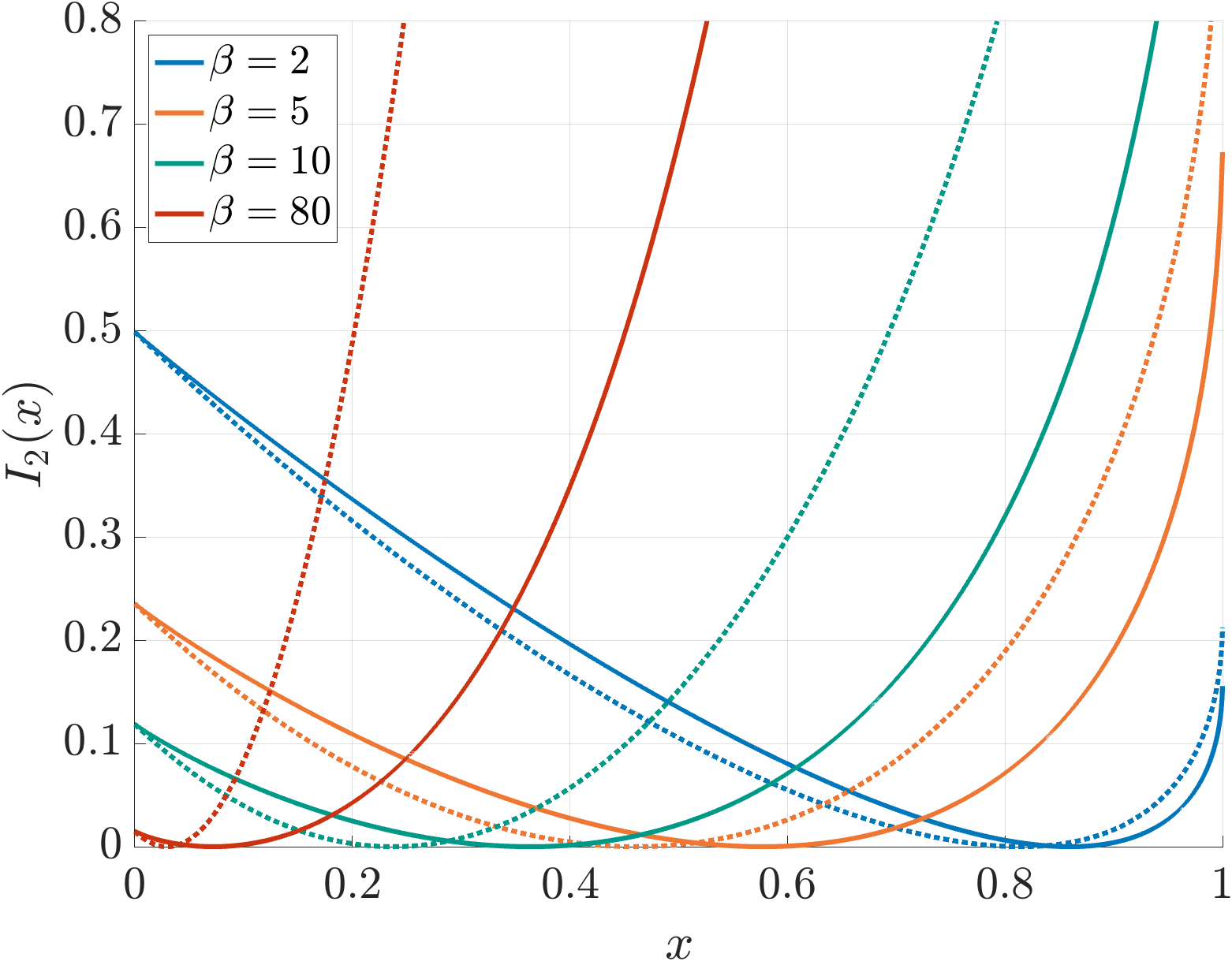}\label{fig:rate}}%
\hfill}
\caption{Functions related to the LDPs (in bits) for the hard-decision optimal
guesswork process $\Bqty{n^{-1}\log_2 \GH(Z^n)}$ (solid) and the soft-decision
optimal guesswork process $\Bqty{n^{-1}\log_2 \GT(N^n)}$ (dashed) in the LRC.}
\label{fig:scgf-rate}
\end{figure*}

\subsection{Scaled Cumulant Generating Functions}

The derivation of the soft-decision sCGF is straightforward and readily follows
from straightforward manipulations. For concision, we generally suppress the
dependence of quantities such as the sCGF on $\beta$, only making it
explicit in the underlying expressions.

\begin{theorem}\label{thm:scgf-soft}
Let $N^n$ be the soft-decision noise effect in the LRC with parameter $\beta$
The sCGF of the soft-decision guesswork process $\Bqty{n^{-1}\ln \GT(N^n)}$ is
\begin{equation*}
    \Lambda_N(\alpha) = \begin{cases}
        \alpha H_{\frac{1}{1+\alpha}}(N) & \alpha\in(-1,\infty), \\
        -H_{\up{min}}(N) & \alpha\leq -1,
    \end{cases}
\end{equation*}
where the R\'enyi entropy rate of $N^n$ is
\begin{equation*}
    \alpha H_{\frac{1}{1+\alpha}}(N) = (1+\alpha) J\pqty{1;
    \frac{\beta}{1+\alpha}} - J(1;\beta)
\end{equation*}
and the min-entropy rate of $N^n$ is
\begin{equation*}
    -H_{\up{min}}(N) = -J(1;\beta).
\end{equation*}
\end{theorem}

\begin{proof}
Theorem 1 and Proposition 4 in \cite{Ari96} imply that
\begin{equation*}
\begin{split}
    (1+\ln 2)^{-\alpha} \exp(\alpha H_{\frac{1}{1+\alpha}}(N^n))
    &\leq \Exp{\upe^{\alpha\ln \GT(N^n)}} \\
    &\leq \exp(\alpha H_{\frac{1}{1+\alpha}}(N^n)).
\end{split}
\end{equation*}
Taking the logarithm and passing to the scaled limit,
\begin{align}
\nonumber
    \lim_{n\to\infty} \frac{1}{n} \ln\Exp{\upe^{\alpha\ln \GT(N^n)}}
    &= \alpha \lim_{n\to\infty} \frac{1}{n} H_{\frac{1}{1+\alpha}}(N^n) \\
\label{eq:lambda-renyi}
    &= \alpha H_{\frac{1}{1+a}}(N).
\end{align}
Substituting the soft-decision PMF for the LRC (\cref{lm:soft-pmf}),
\begin{equation*}
\begin{split}
    \sum_{x}\Pr{N^n=x}^{\frac{1}{1+\alpha}} = \prod_{i=1}^n &\left[
    \pqty{
        \frac{\upe^{\beta i/n}} {1+\upe^{\beta i/n}}
    }^{\frac{1}{1+\alpha}} \right. \\
    &\quad\left. + \pqty{
        \frac{1}{1+\upe^{\beta i/n}}
    }^{\frac{1}{1+\alpha}}
    \right].
\end{split}
\end{equation*}
This implies that
\begin{equation*}
\begin{split}
    \alpha H_{\frac{1}{1+\alpha}}(N^n) &= \sum_{i=1}^n
    (1+\alpha)\ln(1+\upe^{\frac{\beta i}{(1+\alpha) n}}) \\
    &\quad- \sum_{i=1}^n\ln(1+\upe^{\beta i/n}).
\end{split}
\end{equation*}
Scaling by $1/n$ then yields a pair of Riemann sums, both of which converge to
the corresponding integrals, such that
\begin{equation} \label{eq:lambda}
    \alpha H_{\frac{1}{1+a}}(N) =
    (1+\alpha) J\pqty{1;\frac{\beta}{1+\alpha}} - J(1;\beta).
\end{equation}
For all $\beta>0$ and $\alpha>-1$, the integrands in \cref{eq:lambda} are
continuous and finite over $[0,1]$, and hence the integrals are also
well-defined and finite.

Lemma 1 in \cite{CD12} states that if $\Lambda_N(\alpha)$ takes the form of
\cref{eq:lambda-renyi} for all $\alpha>-1$ and it has a continuous derivative
over that range, then $\Lambda_N(\alpha)$ for all $\alpha\leq -1$ is given by
\begin{equation*}
    \Lambda_N(\alpha) = \lim_{n\to\infty} \frac{1}{n} \ln\Pr{\GT(N^n)=1} =
    -H_{\up{min}}(N).
\end{equation*}
We now show that $\Lambda_N'(a)$ exists and is indeed continuous for
$\alpha>-1$. Denoting the integrand of \cref{eq:lambda} by $f(\alpha,x)$,
\begin{equation}\label{eq:fdiv}
    \pdv{\alpha} f(\alpha,x) = \ln(1+\upe^{\frac{\beta x}{1+\alpha}}) -
    \frac{\beta x \upe^{\frac{\beta x}{1+\alpha}}}
    {(1+\alpha)\pqty{1+\upe^{\frac{\beta x}{1+\alpha}}}}.
\end{equation}
Each term in \cref{eq:fdiv} is composition of exponential and logarithmic
functions with positive arguments, so $\pdv{\alpha} f(\alpha,x)$ is continuous.
For any fixed $\alpha>-1$, each term is bounded by a constant over $x\in[0,1]$.
In particular,
\begin{align*}
    \abs{\ln(1+\upe^{\frac{\beta x}{1+\alpha}})} &\leq
    \ln(1+\upe^{\frac{\beta}{1+\alpha}}), \\
    \abs{\frac{\beta x \upe^{\frac{\beta x}{1+\alpha}}}
    {(1+\alpha)\pqty{1+\upe^{\frac{\beta x}{1+\alpha}}}}} &\leq
    \frac{\beta}{1+\alpha}.
\end{align*}
\Cref{eq:fdiv} is thus differentiable over $x\in[0,1]$ for $\alpha>-1$ fixed.
By the dominated convergence theorem, we then obtain
\begin{equation*}
    \Lambda_N'(\alpha) = \int_0^1\pdv{\alpha}f(\alpha,x)\dd{x},
\end{equation*}
which is necessarily continuous for all $\alpha>-1$.

All that remains is to show that
\begin{equation*}
    -H_{\up{min}}(N) = -\int_0^1\ln(1+\upe^{-\beta x})\dd{x}.
\end{equation*}
Since the single most probable noise effect is $0^n$,
\begin{align*}
    -H_{\up{min}}(N) &= \lim_{n\to\infty} \frac{1}{n} \ln(
        \prod_{i=1}^n\frac{1}{1+\upe^{-\beta i/n}} ) \\
    &= - J(1;\beta).
\end{align*}
Again, for $x\in[0,1]$, the integrand is continuous and bounded. Thus,
$-H_{\up{min}}(N)$ is finite and strictly negative.
\end{proof}

The derivation of the hard-decision sCGF is significantly more involved,
although the final expression is wieldy. We defer the proof of the following
theorem to \cref{apx:scgf-hard}.

\begin{theorem}\label{thm:scgf-hard}
Let $Z^n$ be the hard-decision noise effect in the LRC with parameter $\beta$
The sCGF of the hard-decision guesswork process $\Bqty{n^{-1}\ln \GH(Z^n)}$ is
\begin{equation*}
    \Lambda_Z(\alpha) = \begin{cases}
        \alpha H_{\frac{1}{1+\alpha}}(Z) & \alpha\in(-1,\infty), \\
        -H_{\up{min}}(Z) & \alpha\leq -1,
    \end{cases}
\end{equation*}
where the R\'enyi entropy rate of $Z^n$ is
\begin{align*}
\begin{split}
    \alpha H_{\frac{1}{1+\alpha}}(Z) &= \max_{t\in[0,1]} \Bqty{
        \alpha h(t) + J(r_{t,\beta};\beta) - t\ln r_{t,\beta}
    } \\
    &\qquad - J(1;\beta),
\end{split}\\
    r_{t,\beta} &= \frac{\upe^{\beta t} - 1}{1 - \upe^{\beta(t-1)}},
\end{align*}
and the min-entropy rate of $Z^n$ is
\begin{equation*}
    -H_{\up{min}}(Z) = -J(1;\beta).
\end{equation*}
\end{theorem}

The additional complexity in evaluating the hard-decision sCGF is mainly due to
the fact that the R\'enyi entropy of $Z^n$ is not readily expressible as a
simple Riemann sum with a straightforward limit. Nonetheless, the resulting
sCGF has some notable structural similarities, and comparing the two functions
offers one perspective on the differences between hard- and soft-decision ML
decoding in the LRC.

First, note that $H_{\up{min}}(N) = H_{\up{min}}(Z) = -J(1;\beta)$, which
follows from fact that the single most probable noise effects is the same
regardless of whether the reliability ordering permutation is known. We will,
in subsequent sections, again see this property reflected in the fact that the
hard- and soft-decision rate functions agree at $x=0$ (\cref{fig:rate}) and the
success exponents agree at $R=1$ (\cref{fig:expts}). The appearance of a
$-J(1;\beta)$ term in both sCGFs over $\alpha>-1$ is due to the fact that the
two noise effect PMFs can be written with the same normalizer, and is in line
with our expectation that the sCGF be continuous at $\alpha=-1$. As discussed
in \cite{CD12,DLM19}, a discontinuity at $\alpha=-1$ would capture any
exponential growth of the set of most probable noise effects, which does not
grow in the LRC.

The simplicity of the first term in the soft-decision sCGF over $\alpha>-1$
compared to the hard-decision sCGF is primarily due to the fact that when the
the received bits are independent given the reliability ordering permutation.
This allows the R\'enyi entropy to be expressed as an average bit-entropy,
which in the limit is given by an integral. On the other hand, the received
bits are not independent in the hard-decision case.

From the hard-decision perspective, all noise effects with the same Hamming
weight are equiprobable. The R\'enyi entropy is thus given by a sum over the
Hamming weight of possible noise effects, with each term composed of two
factors. The first is a binomial coefficient, which counts sequences of a given
Hamming weight, and second is an elementary symmetric polynomial in terms of
the LRC bit-flip probabilities. This polynomial captures all possible
underlying probabilities for sequences of a given Hamming weight by considering
each possible reliability ordering permutation. In the limit, the binomial
coefficient gives rise to the binary entropy term in the sCGF, while the
elementary symmetric polynomial gives rise to the terms involving
$r_{t,\beta}$. Intuitively, the parameter $r_{t,\beta}$ quantifies the
asymptotically dominant term in the elementary symmetric polynomial for
sequences of Hamming weight $\floor{tn}$. In other words, it is a
parameterization of the most probable sequence of a given Hamming weight, and
this is given by what is essentially a saddle-point optimization. The outer
maximization over $t\in[0,1]$ is then given by a second saddle-point
optimization (specifically, an application of Laplace's method) which picks out
the asymptotically dominant Hamming weight $\floor{tn}$ in the overall sum for
the R\'enyi entropy. See \cref{apx:scgf-hard} for more detail.


As $\beta$ goes to 0, note that both base-2 sCGFs $\Lambda_{N,2}$ and
$\Lambda_{Z,2}$ tend toward (\cref{fig:scgf})
\begin{equation*}
    \Lambda_2(\alpha) = \begin{cases}
        \alpha & \alpha > -\ln 2, \\
        -\ln 2 & \alpha \leq -1,
    \end{cases}
\end{equation*}
which is the sCGF for the guesswork process in a BSC with bit-flip probability
$p=1/2$ \cite{DLM19}. The noise effect distribution of that BSC is also the
limit of both the hard- and soft-decision noise effect distributions in the LRC
as $\beta\to0$.

The following lemma shows that the R\'enyi entropy rate of the soft-decision
noise effect is strictly smaller than that of the hard-decision effect,
except at $\alpha=0$, where they are equal.

\begin{lemma}\label{lm:renyi-order}
Let $N^n$ and $Z^n$ be the soft- and hard-decision noise effects in the LRC
with parameter $\beta$. Then, for all $\alpha>0$, $H_{\alpha}(N) <
H_{\alpha}(Z)$.
\end{lemma}

\begin{proof}
We show that the PMF of $Z^n$ is strictly majorized by the PMF of $N^n$. Since
both the R\'enyi entropy and the Shannon entropy are Schur-concave, this
implies that $H_{\alpha}(N) < H_{\alpha}(Z)$ for all $\alpha>0$.

Let $p_Z,p_{N\gib \tau}\in\R^n$ be non-increasing vectors corresponding to
the PMFs of $Z^n$ and $N^n$ given $T=\tau$. (Note that the $i$th element
of $p_Z$ and $p_{N\gib \tau}$ need not correspond to the same binary sequence,
i.e., they are sorted independently.) Without loss of generality, assume that
$T=\tau_0$. For all $\tau\in S_n$, there exists some permutation matrix
$M_{\tau}$ such that $p_{N\gib\tau} = p_{N\gib \tau_0} M_{\tau}$. Thus,
\begin{equation*}
    p_{Z} = \frac{1}{n!}\sum_{\tau\in S_n} p_{N\gib \tau_0} M_{\tau}
    = p_{N\gib \tau} M,
\end{equation*}
where
\begin{equation*}
    M = \frac{1}{n!}\sum_{\tau\in S_n} M_{\tau}
\end{equation*}
is a doubly stochastic matrix. A result of Hardy, Littlewood, and P\'{o}lya
\cite{HLP29} establishes that this is necessary and sufficient for $p_{Z}$ to
be majorized by $p_{N\gib \tau_0}$.

Letting $p_Z(i)$ and $p_{N\gib \tau_0}(i)$ denote the $i$th element of those
vectors, assume (again without loss of generality) that $p_Z(2)$ and $p_{N\gib
\tau_0}(2)$ are the probabilities for the sequence $x\in\Bqty{0,1}^n$ for which
$x_1=1$ and $x_i=0$ otherwise. This sequence is more probable under
$p_{N\gib\tau_0}$. Since $p_Z(1)$ and $p_{N\gib \tau_0}(1)$ both correspond to
the all-zero sequence,
\begin{equation*}
    \sum_{i=1}^2 p_{N\gib \tau_0}(i) > \sum_{i=1}^2 p_{Z}(i).
\end{equation*}
and the majorization of $p_Z$ by $p_{N\gib \tau_0}$ is strict.
\end{proof}

An immediate corollary is a strict ordering on the sCGFs.

\begin{corollary}\label{pr:scgf-order}
Let $\Lambda_N$ and $\Lambda_Z$ be the sCGFs for soft- and hard-decision
guesswork in the LRC. Then,
\begin{align*}
    \Lambda_{Z}(\alpha) &> \Lambda_{N}(\alpha), \qquad
        \forall\alpha\in(0,\infty), \\
    \Lambda_{Z}(\alpha) &< \Lambda_{N}(\alpha), \qquad \forall\alpha\in(-1,0).
\end{align*}
\end{corollary}


\subsection{Rate Functions}

Given the sCGFs of \cref{thm:scgf-soft,thm:scgf-hard}, it follows that the
guesswork processes satisfy LDPs \cite[Theorem 3]{CD12}.

\begin{theorem}\label{thm:rate}
In the LRC, the soft- and hard-decision guesswork processes both satisfy LDPs
with convex, lower semicontinuous rate functions $I_N,I_Z:[0,\ln 2] \to
[0,\infty)$ given by the Legendre-Fenchel transforms of the sCGFs $\Lambda_N$
and $\Lambda_Z$,
\begin{align*}
    I_N(x) &= \sup_{\alpha\in\R} \Bqty{x\alpha - \Lambda_N(\alpha)}, \\
    I_Z(x) &= \sup_{\alpha\in\R} \Bqty{x\alpha - \Lambda_Z(\alpha)}.
\end{align*}
Furthermore, $I_N$ and $I_Z$ have the following properties, stated in terms of
$N$ but holding identically for $Z$.
\begin{enumerate}
\item $I_N(0) = H_{\up{min}}(N)$.
\item $I_N(x) = 0$ if and only if $x=H_1(N)$.
\item $I_N(x)$ is strictly convex.
\end{enumerate}
\end{theorem}

The rate functions are more readily interpretable than the sCGFs in terms of
decoding behavior, as they describe the asymptotic decay of the probability
that the true noise effect appears at any given position in the ML guessing
function. Thus, the fact that $I_N(x)=0$ if and only if $x=H_1(N)$ implies that
the only position at which the true noise effect appears with non-decaying
probability is growing like $\upe^{nH_1(N)}$ under soft-decision guesswork (and
likewise for hard-decision guesswork and $Z$). In other words, $H_1(N)$ is the
asymptotically ``typical'' value of $n^{-1}\ln \GT(N^n)$. Note that this does
not imply that the mean of $\GT(N^n)$ is growing exponentially with asymptotic
rate $H_1(N)$. Indeed, the asymptotic exponential growth rate of the mean is
given by $\Lambda_N(1) = H_{1/2}(N) \geq H_1(N)$, which was first observed by
Arikan \cite{Ari96}. This distinction is due to the ``long tail'' of guesswork.
Intuitively, the number of possible sequences is growing rapidly in $n$, but
the bulk of the probability is limited to a set sequences which is not growing
so rapidly. This is effectively the phenomenon described by Massey, who showed
and stated that ``there is no interesting upper bound'' on the average
guesswork in terms of the Shannon entropy \cite{Mas94}.

\Cref{lm:renyi-order} implies that $H_1(N) < H_1(Z)$, which is equivalent to
stating that the hard-decision capacity is less than the soft-decision
capacity. This is reflected in the fact that the zero of the soft-decision rate
function is always less than the zero of the hard-decision rate function
(\cref{fig:rate}). The following result states the soft- and hard-decision
rate functions are strictly ordered outside of the interval $(H_1(N),H_1(Z))$.

\begin{proposition}\label{pr:rate-order}
Let $I_N$ and $I_Z$ be the rate functions for soft- and hard-decision guesswork
in the LRC. Then,
\begin{align*}
    I_N(x) &< I_Z(x), \qquad 0 < x \leq H_1(N), \\
    I_N(x) &> I_Z(x), \qquad H_1(Z) \leq x.
\end{align*}
\end{proposition}

\begin{proof}
Since $\Lambda_N$ and $\Lambda_Z$ are strictly convex over $[-1,\infty)$,
\begin{align*}
    I_N(x) &= x k(x) - \Lambda_N(k(x)), \\
    I_Z(x) &= x l(x) - \Lambda_Z(l(x)),
\end{align*}
where $k(x),l(x)\in[-1,\infty)$ are the unique points for which
\begin{equation*}
    \Lambda'_N(k(x))=x, \qquad \Lambda'_Z(l(x))=x.
\end{equation*}

When the sCGF is strictly convex, the duality property of the Legendre-Fenchel
transform states that the slope of the sCGF at 0 is the point at which the
slope of the rate function is 0, i.e., $\Lambda'_N(0) = H_1(N)$ and
$\Lambda'_Z(0) = H_1(Z)$. Recalling that $H_1(N) < H_1(Z)$,
\begin{align*}
    x<H_1(N) &\implies k(x),l(x) < 0, \\
    x>H_1(Z) &\implies k(x),l(x) > 0.
\end{align*}

By \cref{pr:scgf-order}, $\Lambda_Z(\alpha) < \Lambda_N(\alpha) < 0$ for
$\alpha\in(-1,0)$. This implies that
\begin{equation*}
    I_N(x) = \sup_{\alpha \in R} \Bqty{x \alpha - \Lambda_N(\alpha)} <
    \sup_{\alpha \in R} \Bqty{x \alpha - \Lambda_Z(\alpha)} = I_Z(x)
\end{equation*}
if both maximizers lie in $(-1,0)$, i.e., if $0 < x < H_1(N)$.

Similarly, $0 < \Lambda_N(\alpha) < \Lambda_Z(\alpha)$ for
$\alpha\in(0,\infty)$. Thus,
\begin{equation*}
    I_N(x) = \sup_{\alpha \in R} \Bqty{x \alpha - \Lambda_N(\alpha)} >
    \sup_{\alpha \in R} \Bqty{x \alpha - \Lambda_Z(\alpha)} = I_Z(x)
\end{equation*}
if both maximizers lie in $(0,\infty)$, i.e., if $H_1(Z) < x$.

To see that the established inequalities hold for $x = H_1(N)$ and $x=H_1(Z)$
respectively, it suffices to note that although one maximizer is 0, the other
remains in the desired range.
\end{proof}

The following proposition bounds the slope of the rate functions.

\begin{proposition}\label{pr:scgf-slope}
Let $\Lambda_N$ and $\Lambda_Z$ be the sCGFs for soft- and hard-decision
guesswork in the LRC. Then, $\Lambda'_N(\alpha)\in (0,\ln 2)$ and
$\Lambda'_Z(\alpha) \in (0, \ln2)$ for $\alpha\in(0,\infty)$.
\end{proposition}

\begin{proof}
The R\'{e}nyi entropy (of any order) is at most $H_0$. Thus,
\begin{equation*}
    H_{\frac{1}{1+\alpha}}(Z^n) < n\ln 2,
\end{equation*}
where the inequality is strict because  $Z^n$ is not distributed uniformly.
This implies that $0 < \Lambda_Z(\alpha) < \alpha\ln 2$ for $\alpha>0$, and
thus that $\Lambda'_Z(\alpha) \in (0,\ln 2)$ over that same range. By
\cref{pr:scgf-order}, $\Lambda_Z(\alpha) > \Lambda_N(\alpha)$ for
$\alpha\in(0,\infty)$, and so the same bound can be applied to $\Lambda_N$.
\end{proof}

Because the slopes of the sCGFs never reach $\ln 2$, it follows from the
duality of the Legendre-Fenchel transform that the rate functions diverge at
$x=\ln 2$ (in bits, the rate functions diverge at $x=1$, as seen in
\cref{fig:rate}). Physically, this means that the noise effect appears near the
very end of the optimal guessing order with a probability that is decaying
incredibly fast as the block length growths, which is to be expected. Note that
this divergence does not occur at $x=0$, because the slope of the sCGF is
indeed 0 at $\alpha=-1$.

\section{Error Exponents for\texorpdfstring{\\}{ }Hard- and Soft-Decision
Decoding}\label{sec:error}

\begin{figure*}
\centering
{\hfill%
\subfloat[low $\beta$]
    {\includegraphics[width=\imgwidth]{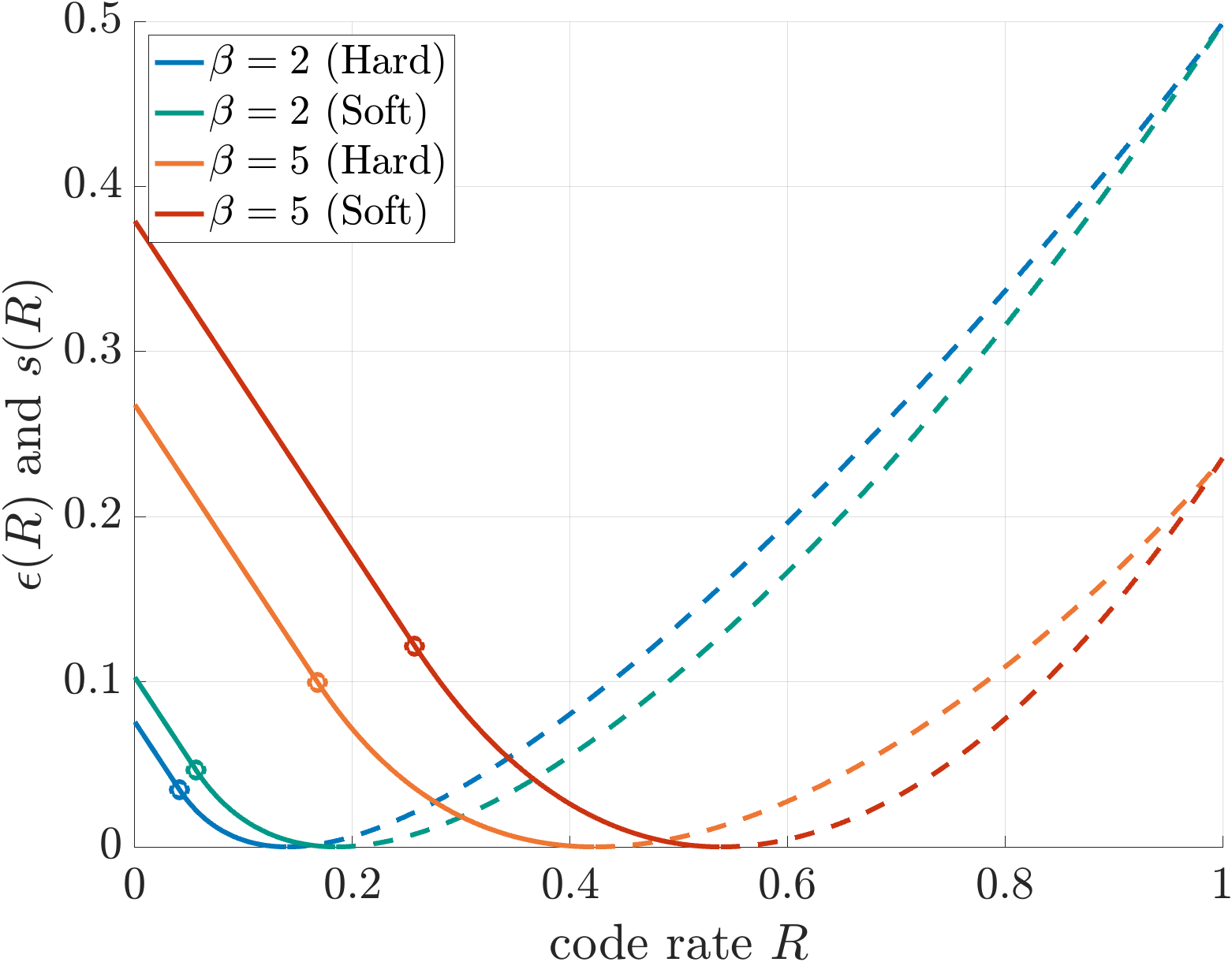}}%
\hfill%
\subfloat[high $\beta$]
    {\includegraphics[width=\imgwidth]{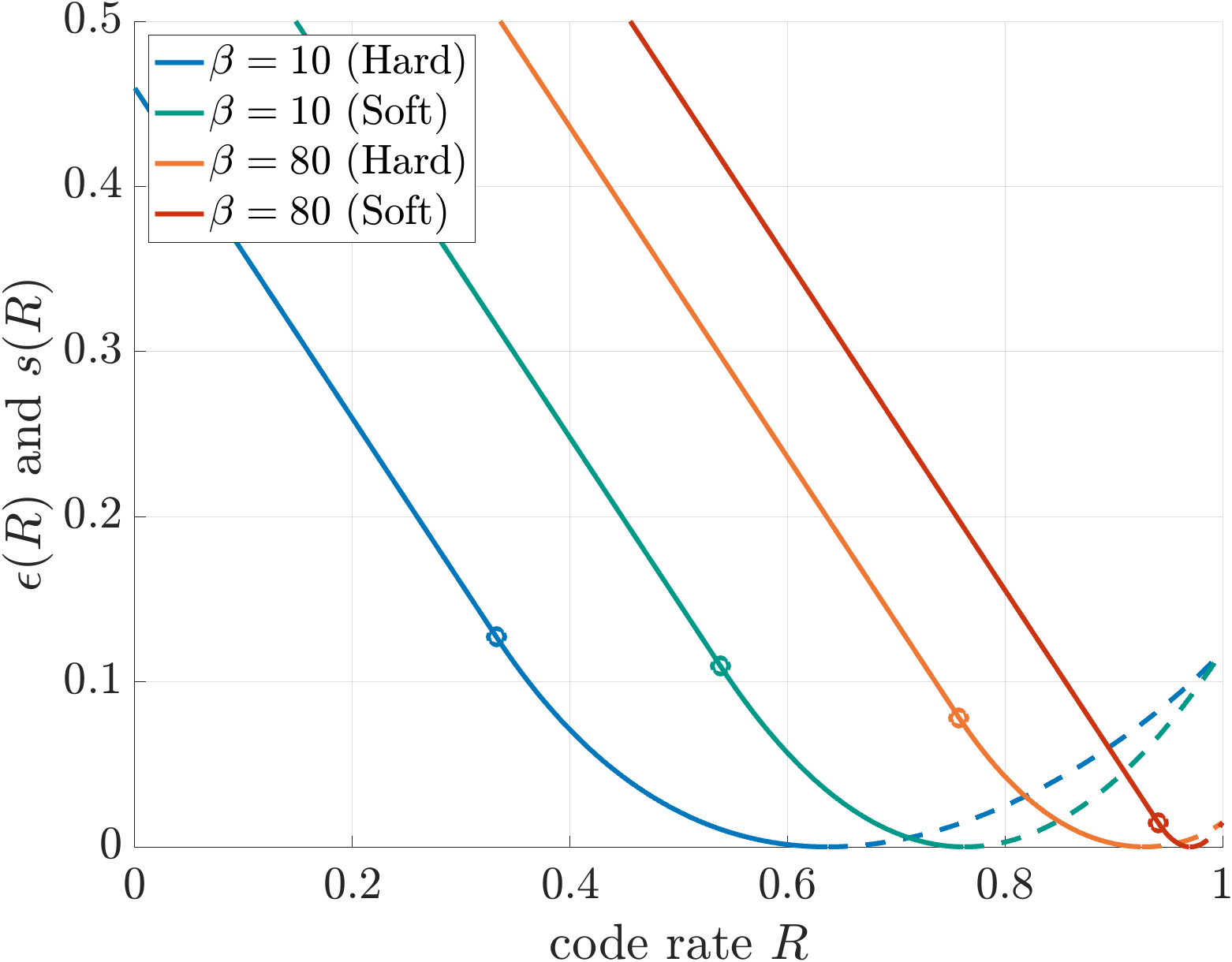}}%
\hfill}
\caption{Error exponents $\epsilon(R)$ (solid) and success exponents $s(R)$
(dashed) for soft- and hard-decision ML decoding in the LRC. Circles mark the
critical rate $R_{\up{cr}}$ at which $\epsilon(R)$ transitions from being
linear to strictly convex.}
\label{fig:expts}
\end{figure*}

As discussed in \cref{sec:ldp}, a decoding error occurs whenever the true noise
effect $N^n$ appears later in the guessing order than the first spurious noise
effect $U^n$ which also yields a code word when added to the received sequence.
The probability of a decoding error is thus
\begin{equation*}
    \Pr{G(U^n) < G(N^n) } = \Pr{n^{-1}\ln(\frac{G(U^n)}{G(N^n)}) < 0}
\end{equation*}
Given that $\Bqty{n^{-1}\ln G(N^n)}$ and $\Bqty{n^{-1}\ln G(U^n)}$ are
independent processes both satisfying LDPs with rate functions $I_N$ and $I_U$,
the joint process $\Bqty{n^{-1}\ln G(N^n),n^{-1}\ln G(U^n)}$ also satisfies an
LDP with rate function $I_{N,U}(x,y) = I_N(x) + I_U(y)$. The contraction
principle \cite{DZ10} states that applying a continuous function $f$ to a
process satisfying an LDP results in a new process also satisfying an LDP, with
the new rate function given by a transformation corresponding to the applied
function. Taking $f(x,y)=x-y$, the process $\Bqty{n^{-1}\ln(G(U^n)/G(N^n))}$
satisfies an LDP with rate function
\begin{equation*}
    I_{U/N}(x) = \inf_{a,b}\Bqty{ I_U(a) + I_N(b) : x=a-b }.
\end{equation*}

This is the approach used to prove Proposition 1 in \cite{DLM19}, which states
that GRAND is capacity-achieving while also establishing error and success
exponents, assuming that the code book is sampled uniformly at random. For such
a code, $G(U^n)$ is approximately exponentially distributed and the process
$\Bqty{n^{-1}\ln G(U^n)}$ satisfies an LDP \cite[Theorem 2]{DLM19}, where $G$
is any guessing function. We restate this result here in our notation. Note
that, when working in bits and thus with all logarithms taken to base-2, the
sCGF of a general process $A^n$ is defined to be
\begin{equation*}
    \Lambda_{A,2}(\alpha) = \lim_{n\to\infty} \frac{1}{n} \log_2
    \Exp{2^{\alpha n A^n}}.
\end{equation*}
Then, $\Lambda_{A,2}$ and the corresponding base-2 rate function $I_{A,2}$ are
given by the following transformations of the base-$\up{e}$ functions:
\begin{align}
\label{eq:b2-scgf}
    \Lambda_{A,2}(\alpha) &= \frac{1}{\ln 2}\Lambda_A(\alpha), \\
\label{eq:b2-rate}
    I_{A,2}(\alpha) &= \frac{1}{\ln 2} I_{A}(\alpha \ln 2).
\end{align}

\begin{theorem}[\cite{DLM19}]\label{thm:err-exp}
Let $N^n$ be a channel noise effect process, let $G$ be the optimal guessing
function for $N^n$, and assume that the guesswork process $\Bqty{n^{-1}\log_2
G(N^n)}$
satisfies an LDP with (base-2) rate function $I_{N,2}$. Let $R\in(0,1)$ be the
code rate and assume that the code book is sampled uniformly at random. Let
$U^n$ denote the first incorrect noise effect guessed by $G$ which also
corresponds to a code word.

If the code rate is below the channel capacity, i.e., if
\begin{equation*}
    R < C_N = 1-H_1(N),
\end{equation*}
then the probability that the GRAND algorithm using guessing function $G$ fails
to identify the transmitted code word decays exponentially in the block length
$n$. In particular,
\begin{align*}
    \epsilon(R) &= -\lim_{n\to\infty} \frac{1}{n}
        \log_2 \Pr{G(U^n) < G(N^n)} \\
    &= \begin{cases}
        1 - R - H_{1/2}(N) & R\in(0,1-x^{\ast}) \\
        I_{N,2}(1-R) & R\in [1-x^{\ast}, C_N),
    \end{cases}
\end{align*}
where $x^{\ast}\in[0,1]$, which is assumed to exist, is given by
\begin{equation*}
    I'_{N,2}(x^{\ast}) = 1.
\end{equation*}
Furthermore, the probability of a correct decoding does not decay exponentially
in $n$, i.e.,
\begin{equation*}
    s(R) = -\lim_{n\to\infty} \frac{1}{n} \log_2 \Pr{G(U^n) > G(N^n)} = 0.
\end{equation*}

Alternatively, if the code rate is above the channel capacity, the probability
of a correct decoding does decay exponentially in $n$, while the probability of
a decoding error does not. In particular, $s(R) = I_{N,2}(1-R)$ and
$\epsilon(R)=0$.
\end{theorem}

The transition point $1-x^{\ast}$, below which the error exponent is linear and
above which it is strictly convex, was first observed by Gallager in the
context of discrete-time memoryless channels \cite{Gal65}, who called it the
\emph{critical rate}. In the sequel, we accordingly denote the point
$1-x^{\ast}$ by $R_{\up{cr}}$. The analysis via which Gallager demonstrated the
existence of this critical rate, however, does not illuminate why the error
exponent is linear in one regime and strictly convex in the other. The large
deviations approach via GRAND offers a clear interpretation. At rates below
$R_{\up{cr}}$, the most likely way for a decoding error to occur is that
$G(N^n)$ is near its average, which is why $H_{1/2}(N)$ appears, but the first
spurious noise effect $U^n$ appears atypically early. At rates above
$R_{\up{cr}}$, the code, and thus $G(U^n)$, are typical, but the noise effect
is exceptionally unlikely and far down in the guessing order, which is why this
portion of the error exponent is given by the rate function $I_N$.

To apply \cref{thm:err-exp} to the LRC, we need only show that the critical
rate exists under both hard- and soft-decision guesswork. The following
proposition does so, and further shows that the critical rate is greater for
soft-decision guesswork than it is for hard-decision.

\begin{proposition}\label{pr:rcr}
Let $I_{N,2}$ and $I_{Z,2}$ be the rate functions (in bits) for soft- and
hard-decision guesswork in the LRC. There exist unique
$x^{\ast},y^{\ast}\in(0,1)$ such that
\begin{equation*}
    I'_{N,2}(x^{\ast}) = 1, \qquad I'_{Z,2}(y^{\ast}) = 1.
\end{equation*}
Furthermore, $1-y^{\ast} < 1-x^{\ast}$.
\end{proposition}

\begin{proof}
Since $\Lambda_{N,2}$ is strictly convex over $[-1,\infty)$, by the duality of
the Legendre-Fenchel transform, $x^{\ast} = \Lambda'_{N,2}(1)$. By
\cref{pr:scgf-slope}, $\Lambda'_{N,2}(1)\in(0,1)$ and thus $x^{\ast}$ exists
and is unique. The same argument holds for $y^{\ast}$ and $\Lambda'_{Z,2}(1)$.

In \cref{apx:d-order}, and in \cref{lm:scgf-d-order} in particular, it is shown
that $\Lambda'_{N,2}(1) < \Lambda'_{Z,2}(1)$. It follows that $1 - y^{\ast} < 1
- x^{\ast}$.
\end{proof}

The following result gives a strict ordering on the error and success
exponents, showing that soft-decision ML decoding outperforms hard-decision ML
decoding in the LRC.

\begin{proposition}\label{pr:soft-win}
Let $\epsilon_{N}(R)$ and $\epsilon_{Z}(R)$ denote the error exponents for
soft- and hard-decision ML decoding in the LRC. Then, $\epsilon_{N}(R) >
\epsilon_{Z}(R)$ for all $R\in[0,C_Z]$, i.e., when the code rate is below the
hard-decision capacity.

Similarly, let $s_N(R)$ and $s_Z(R)$ denote the respective success exponents in
the LRC. Then, $s_N(R)<s_Z(R)$ for all $R\in[C_N,1)$, i.e., when the code rate
is above the soft-decision capacity.
\end{proposition}

\begin{proof}
Let $R^{\ast}_Z$ and $R^{\ast}_N$ denote the critical rates for hard- and
soft-decision decoding in the LRC respectively. By \cref{pr:rcr}, $R^{\ast}_Z
< R^{\ast}_N$. By \cref{lm:renyi-order}, $C_Z < C_N$.

\Cref{lm:renyi-order} also gives $H_{1/2}(N) < H_{1/2}(Z)$ and thus that
$\epsilon_Z(R) < \epsilon_N(R)$ for $R\in[0,R^{\ast}_Z]$, the regime over which
both exponents are linear. Similarly, it follows from \cref{pr:rate-order} that
$\epsilon_Z(R) < \epsilon_N(R)$ for $R\in[R^{\ast}_N,C_Z]$, over which
$\epsilon_N(R)$ is strictly convex and $\epsilon_Z(R)$ is either strictly
convex or zero. The fact that $\epsilon_Z(R) < \epsilon_N(R)$ over the
intermediate region $R\in[R^{\ast}_Z,R^{\ast}_N]$ then follows from the
convexity of the error exponents.

Finally, \cref{pr:rate-order} implies that $s_N(R) < s_Z(R)$ for $R\in[C_N,1)$.
\end{proof}

\Cref{pr:soft-win} asserts that the error and success exponents for hard- and
soft-decision ML decoding are never identical, but the magnitude of the
difference does depend on $\beta$ (\cref{fig:expts}). Intuitively, when $\beta$
is small, the majority of the bits are unreliable and the exact reliability
ordering permutation does not offer much additional information. In that case,
guessing by Hamming weight is nearly optimal. On the other hand, when $\beta$
is high, most bits are correctly received and the correct noise effect will be
guessed early enough by both decoders such that the difference in performance
is relatively small. The difference is most noticeable in the intermediate
regime, where $\beta$ is big enough for there to be a substantial portion of
reliable bits and knowing the reliability ordering permutation is valuable, but
small enough such that the noise effect is not guessed too early.

\begin{figure*}
\centering
{\hfill%
\subfloat[rate functions (hard-decision LRC vs BSC)]
    {\includegraphics[width=\imgwidth]{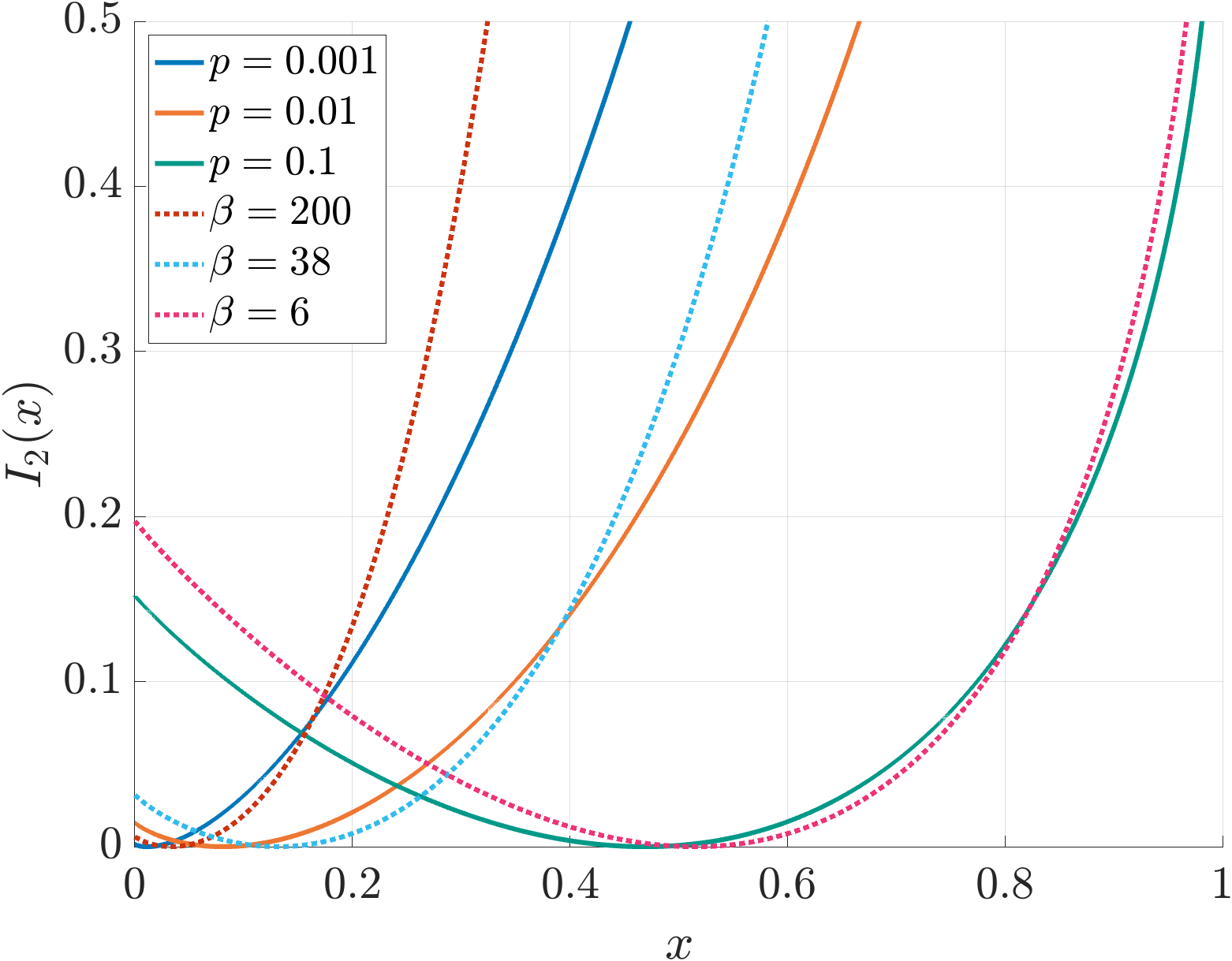}\label{fig:hr}}%
\hfill%
\subfloat[error exponents (hard-decision LRC vs BSC)]
    {\includegraphics[width=\imgwidth]{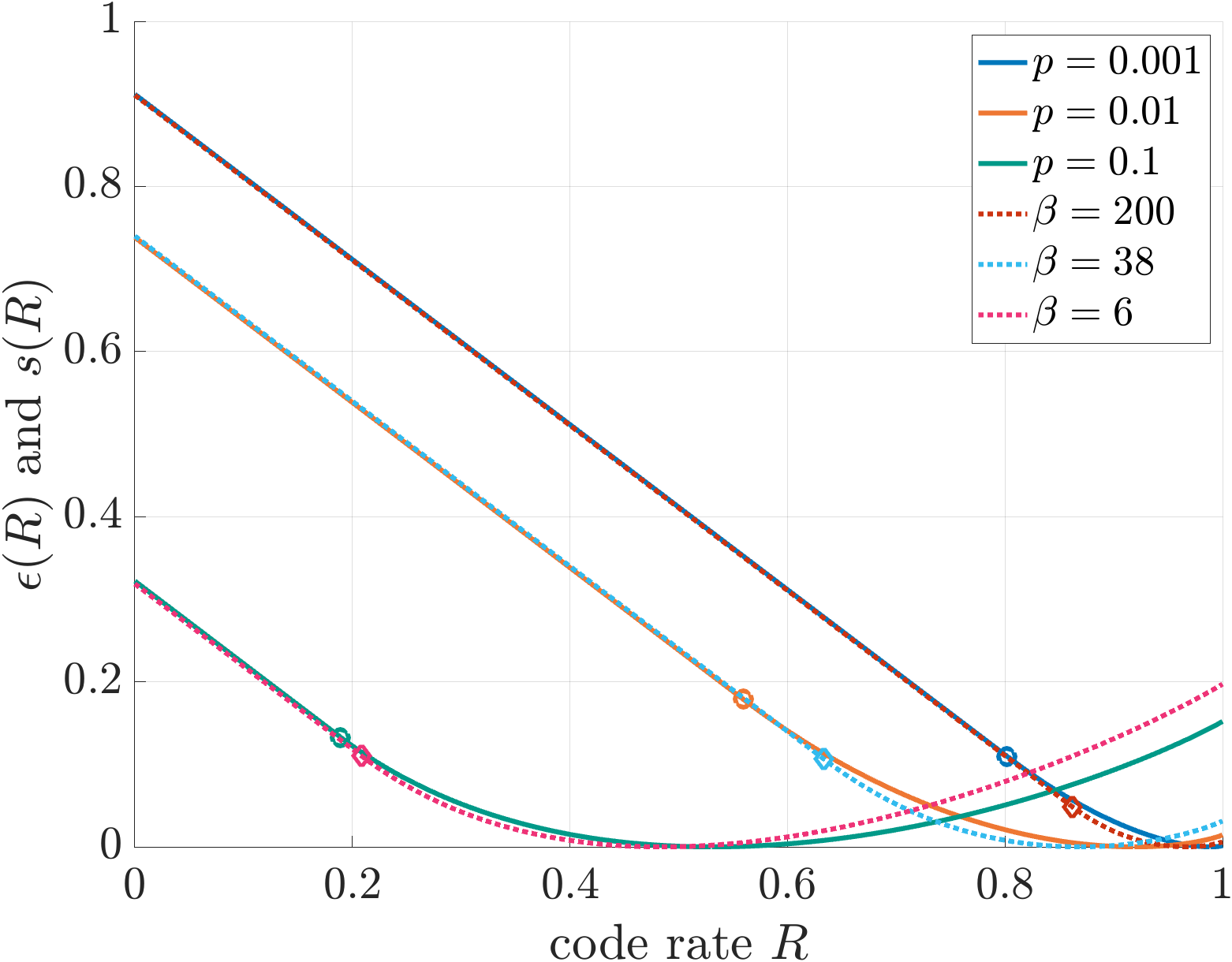}\label{fig:he}}%
\hfill}

{\hfill%
\subfloat[rate functions (soft-decision LRC vs BSC)]
    {\includegraphics[width=\imgwidth]{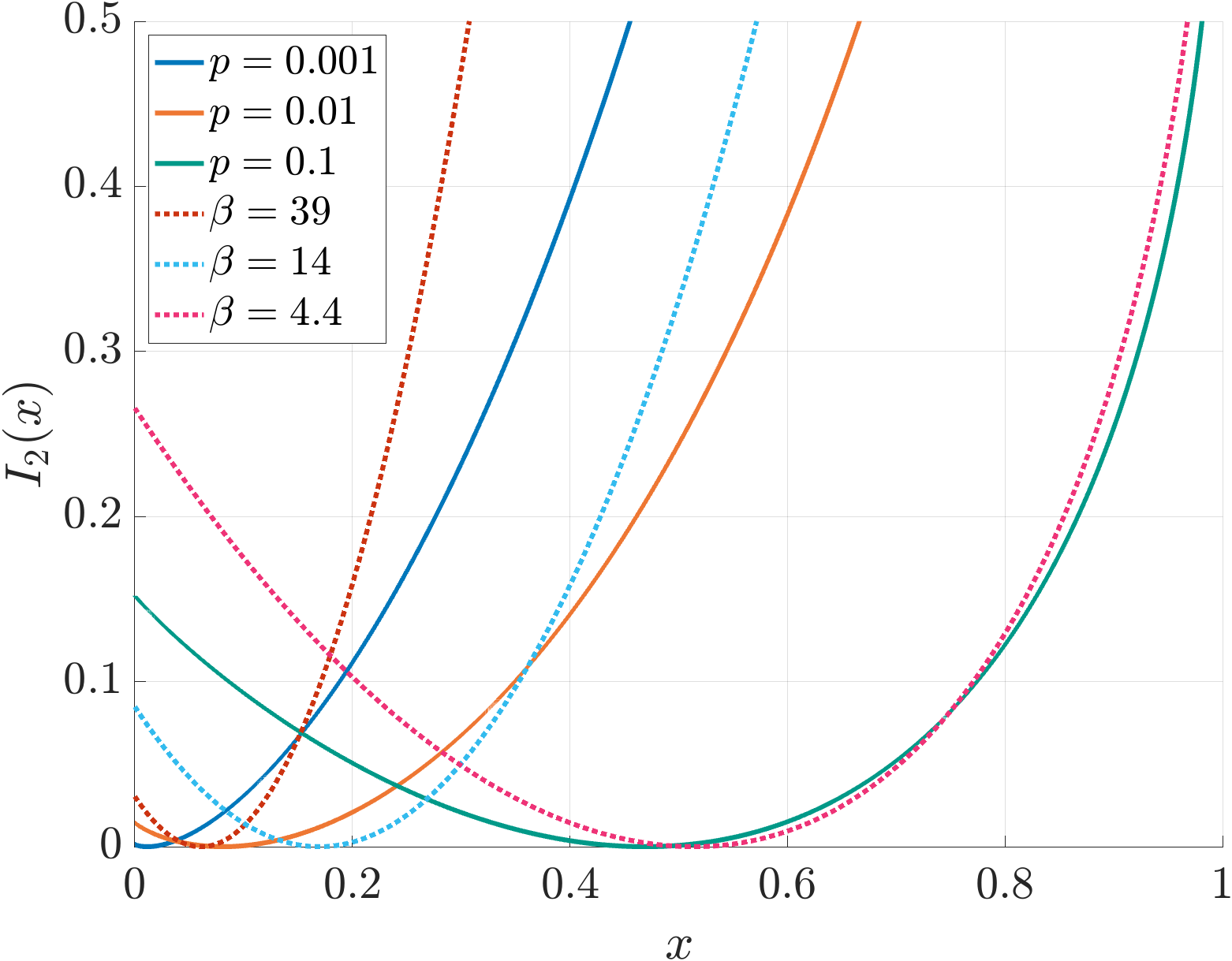}\label{fig:sr}}%
\hfill%
\subfloat[error exponents (soft-decision LRC vs BSC)]
    {\includegraphics[width=\imgwidth]{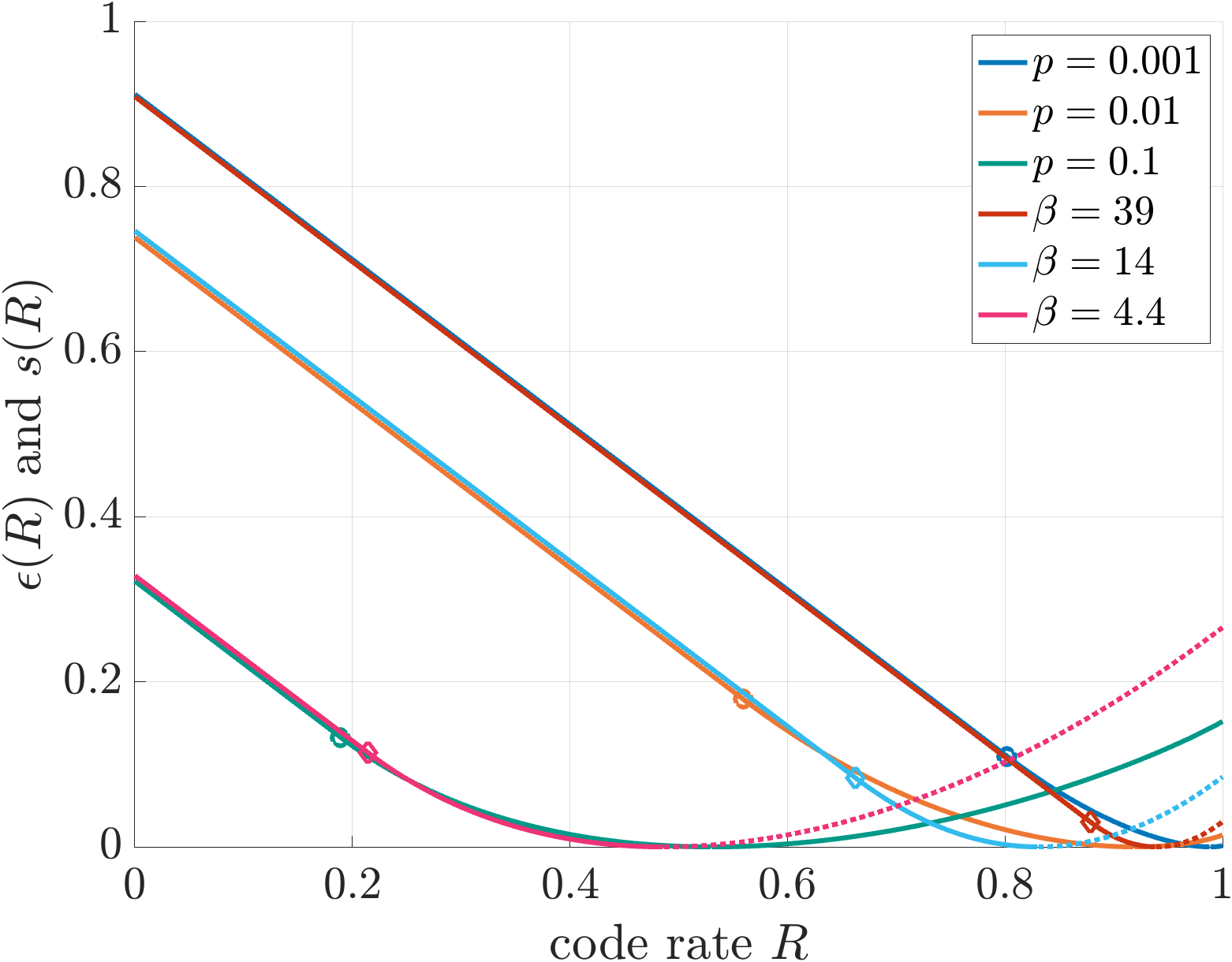}\label{fig:se}}%
\hfill}
\caption{Comparison between rate functions and error exponents for the BSC
(solid lines) and the LRC (dotted lines). For simplicity, the error and
success exponents are plotted together, rather than distinguishing between them
with different line styles, as in \cref{fig:expts}. The particular choices of
$\beta$ were made such that the two channels have the same average guesswork
(given by $H_{1/2}$).
}
\label{fig:bsc-lrc}
\end{figure*}

The error and success exponents for the LRC offer one way of interpreting the
relative noisiness of the channel at a particular value of $\beta$. Below the
critical rate, any channel with the same average guesswork, i.e., any channel
with noise of the same R\'enyi entropy rate $H_{1/2}$, will have the same error
exponents. Using this property to compare the LRC to the BSC
(\cref{fig:bsc-lrc}) gives one heuristic for mapping the LRC parameter $\beta$
to the BSC bit-flip parameter $p$. Roughly speaking, $p$ values which are an
order of magnitude apart correspond to $\beta$ values which are also an order
of magnitude apart in the LRC (\cref{fig:hr,fig:he}). Under soft-decision
decoding, this range is compressed, however, with performance degrading much
more slowly as $\beta$ decreases. (\cref{fig:sr,fig:se}). Naturally, other
ways of matching any pair of channels would lead to different parameter
relationships. Simply matching capacities is one option, but the difference in
the curvature between the rate functions of the two channels implies that the
decoding performance is also potentially very different. Matching the average
guesswork has the benefit of matching the decoding performance, at least over a
particular range of code rates.

\section{Conclusion}\label{sec:conclusion}

We introduced the linear reliability channel, a discrete channel with a
formally analyzable the soft-decision maximum likelihood decoder, and we
established explicit error exponents quantitatively demonstrating the gain in
performance from fully exploiting the channel soft information. Because the LRC
can well-approximate a wide range of continuous-noise channels, further
analysis of the LRC and quantities such as the logistic weight, which are
intimately connected with its soft-decision ML decoder, may point towards
future directions in code construction and coding theory  tailored to a
soft-decision setting.

By extending the large deviations style of analysis originally aimed solely at
hard-decision decoding with GRAND, the LRC highlights the potential of discrete
channels to offer novel insights into continuous channels. Indeed, while the
LRC may be viewed as simply an approximation of channels of real interest, it
offers a unified framework with analytical results which are simple, clean, and
readily interpretable. Because the soft-decision error exponent is computable,
the LRC can also be used as a theoretical benchmark against which the empirical
performance of soft-decision decoding algorithms can be directly evaluated.

The natural emergence of the logistic weight in the LRC has significant
implications for how code quality should be assessed in soft-decision settings.
The difference between the hard-decision and soft-decision error exponents
demonstrates that classical metrics such as the minimum Hamming distance of a
code may not imply good performance when soft information is available. Future
work investigating techniques for soft-decision-centric code construction,
e.g., on the basis of maximizing the minimum Logistic weight of a code, could
offer further fundamental insight into the problem of decoding in the presence
of soft information.

\emergencystretch=1em
\printbibliography

\appendices\crefalias{section}{appendix}
\appendices\crefalias{subsection}{appendix}
\section{Proof of the sCGF for Hard-Decision Guesswork}\label{apx:scgf-hard}

The key ingredient in the proof of \cref{thm:scgf-hard} is the asymptotic
exponential growth rate of the elementary symmetric polynomials $a^n_k(\beta)$.
Recall that $r_{t,\beta}$ is defined in the statement of \cref{thm:scgf-hard}.
The following proposition is based on the coinciding lower (\cref{lm:liminf})
and upper bounds (\cref{lm:limsup}), established in \cref{sec:hard-lower} and
\cref{sec:hard-upper} respectively. The proof of \cref{thm:scgf-hard} itself is
then given in \cref{sec:hard-final}.

\begin{proposition}\label{pr:ank}
For $t\in(0,1)$,
\begin{equation*}
    \lim_{n\to\infty} \frac{1}{n} \ln a^{n}_{\floor{tn}}(\beta) =
    J(r_{t,\beta};\beta) - t\ln r_{t,\beta}.
\end{equation*}
Furthermore, this convergence is uniform over $t\in(0,1)$.
\end{proposition}

\begin{proof}
The bounds of \cref{lm:limsup} and \cref{lm:liminf} together yield the desired
limit. To show uniform convergence, it suffices to note that, since
$\ln(1+r_{t,\beta}\upe^{-\beta x})$ is uniformly continuous over $t\in[0,1]$
and $x\in[0,1]$,
\begin{equation*}
    \lim_{n\to\infty} \sup_{t\in[0,1]} \abs{\frac{1}{n} \sum_{i=1}^n \ln(1 +
    r_{t,\beta}\upe^{-\beta i/n}) - J\pqty{r_{t,\beta};\beta}} = 0.
\end{equation*}
The only other error terms occur in the lower bound, in particular, the
$\ca{O}(\ln(n)/\sqrt{n})$ and $\ca{O}(n^{-1})$ terms which are treated in the
proofs of \cref{lm:e-bound} and \cref{lm:liminf}. These do not depend on $t$
and thus vanish uniformly as $n\to\infty$.
\end{proof}

The proof of \cref{thm:scgf-hard} also makes use of the fact that the function
which is being maximized over $t$ in the expression for the R\'enyi entropy
rate is concave in $t$.

\begin{lemma}\label{lm:hard-concave}
With $\alpha>-1$ and $\beta>0$ fixed, the function
\begin{equation*}
    f_{\alpha,\beta}(t) = \alpha h(t) +
    J\pqty{r_{t,\beta};\beta} - t\ln r_{t,\beta}
\end{equation*}
is strictly concave over $t\in(0,1)$.
\end{lemma}

\begin{proof}
Let
\begin{equation*}
    g(r,t) = J\pqty{r;\beta} - t\ln r.
\end{equation*}
In \cref{sec:hard-upper}, it is shown that $r_{t,\beta}$ is the solution to the
parameterized optimization problem $\min_{r>0} g(r,t)$. By the envelope
theorem \cite{Aki85},
\begin{equation*}
    \dv{t} \bqty{g(r_{t,\beta},t)} = g'_{\ast}(t),
\end{equation*}
where $g_{\ast}(t) = g(r_{t,\beta},t)$. Taking derivatives,
\begin{align*}
    g''_{\ast}(t) = -\beta\pqty{
    \frac{1}{1-\upe^{-\beta t}} + \frac{1}{\upe^{\beta(1-t)} - 1}
    }
    \leq -\frac{1}{t(1-t)}.
\end{align*}
Since the binary entropy function is concave with
\begin{equation*}
    h''(t) = -\frac{1}{t(1-t)} < 0,
\end{equation*}
it follows that
\begin{align*}
    f''_{\alpha,\beta}(t) = \alpha h''(t) + g''_{\ast}(t)
    \leq -\frac{(\alpha+1)}{t(1-t)}
     < 0.
\end{align*}
\end{proof}

\subsection{Lower Bound}\label{sec:hard-lower}

We define the parameterized discrete random variable $K_n(r,\beta)$ with a PMF
depending on $a^n_{k}(\beta)$ and the arbitrary positive constant $r$. An
asymptotic lower bound on $a^n_{\floor{tn}}(\beta)$ is then obtained by
analyzing the mode of $K_n(r,\beta)$ a well-chosen value of $r$.

\begin{definition}
The discrete random variable $K_n(r,\beta)$ with parameters $r,\beta>0$ has PMF
\begin{equation}\label{eq:k-pmf}
    p_{K_n}(k;r,\beta) = \frac{a^n_k(\beta)r^k}{E_n(r;\beta)}, \qquad 0\leq k
    \leq n,
\end{equation}
where $E_n(r;\beta)$ is the normalizing constant
\begin{equation*}
    E_n(r;\beta) = \sum_{k=0}^n a^n_k(\beta)r^k.
\end{equation*}
\end{definition}

We first show that $K_n$ is log-concave.

\begin{definition}
An integer-valued random variable $X$ with PMF $p_X$ is \emph{log-concave} if,
for all $x\in\Z$,
\begin{equation*}
     p_X(x+1)p_X(x-1) \leq p_X(x)^2.
\end{equation*}
If the inequality is strict, $X$ is \emph{strictly log-concave}.
\end{definition}

\begin{lemma}\label{lm:mode}
For all $r,\beta>0$, the discrete random variable $K_n(r,\beta)$ is strictly
log-concave and thus has a unique mode $\kappa_n(r,\beta)
=\max_{k}p_{K_n}(k;r,\beta)$.
\end{lemma}

\begin{proof}
Define for convenience
\begin{equation*}
    f_n(k;r,\beta) = \ln a^n_k(\beta) + k\ln r.
\end{equation*}
Since $a^n_k(\beta)\geq 0$ and $\upe^{-\beta i/n}\neq\upe^{-\beta j/n}$ for
$i\neq j$, Newton's inequalities \cite{HLP52}
yield
\begin{equation}\label{eq:newton}
    a^n_{k-1}(\beta)a^n_{k+1}(\beta) <
    \frac{\binom{n}{k-1}\binom{n}{k+1}}{\binom{n}{k}^2} \pqty{a^n_k(\beta)}^2
    < \pqty{a^n_k(\beta)}^2
\end{equation}
It follows that
\begin{equation*}
    f_n(k+1;r,\beta) + f_n(k-1;r,\beta) < 2f_n(k;r,\beta),
\end{equation*}
and thus $K_n(r,\beta)$ is log-concave. Because the inequality in
\cref{eq:newton} is strict, $f_n(k;r,\beta)$ has a unique maximum and thus
$K_n(r,\beta)$ has a unique mode.
\end{proof}

Seeking to show sufficient concentration around the unique mode, which is
guaranteed to exist by \cref{lm:mode}, the next lemma describes the
variance of $K_n(r,\beta)$.

\begin{lemma}\label{lm:var}
The variance $\sigma^2_n$ of $K_n(r,\beta)$ is of order $\Theta(n)$.
\end{lemma}

\begin{proof}
We first show that $K_n(r,\beta)$ is equivalently given by the Poisson binomial
distribution, which describes the probability of observing $k$ successes over
$n$ trials when the $i$th trial has success probability
\begin{equation*}
    p_i = \frac{r\upe^{-\beta i/n}}{1+r\upe^{-\beta i/n}} \in (0,1).
\end{equation*}
The PMF of the Poisson binomial is given by
\begin{equation*}
    P(k) = \sum_{A\in F_k} \prod_{i\in A}p_i \prod_{i\notin A}(1-p_i),
\end{equation*}
where $F_k$ is the set of all subsets of $[n]$ of cardinality $k$. Thus,
\begin{align*}
    P(k) &= \frac{r^k}{\prod_{i=1}^n \pqty{1+r\upe^{-\beta i/n}}}
    \sum_{A\in F_k} \upe^{-\beta \sum_{i\in A} i/n} \\
    &= \frac{a^n_k(\beta) r^k}{E_n(r)},
\end{align*}
in agreement with \cref{eq:k-pmf}. The Poisson binomial has variance
\begin{equation*}
    \sigma^2_n = \sum_{i=1}^n (1-p_i)p_i
    = \sum_{i=1}^n \frac{r\upe^{-\beta i/n}}{\pqty{1+r\upe^{-\beta i/n}}^2}
\end{equation*}
Scaling $1/n$ yields a Riemann sum which converges to
\begin{equation*}
    \lim_{n\to\infty} \frac{\sigma^2_n}{n} = \int_0^1
    \frac{r\upe^{-\beta x}}{\pqty{1+r\upe^{-\beta x}}^2} \dd{x}
\end{equation*}
For fixed $r$ and $\beta$, this limit is a non-zero finite constant, and hence
$\sigma^2_n = \Theta(n)$.
\end{proof}

We now show that the mode of $K_n(r,\beta)$ is asymptotically growing linearly
in $n$ and that it approaches this linear limit at rate $\Theta(\sqrt{n})$.

\begin{lemma}\label{lm:mode-limit}
For all $r,\beta>0$,
\begin{equation*}
    \lim_{n\to\infty} \frac{\kappa_n(r,\beta)}{n} =
    \frac{1}{\beta}\ln(\frac{1+r}{1+r\upe^{-\beta}}).
\end{equation*}
Furthermore,
\begin{equation*}
    \abs{\kappa_n(r;\beta) - \frac{n}{\beta}\ln(\frac{1+r}{1+r\upe^{-\beta}})}
    = \Theta(\sqrt{n}).
\end{equation*}
\end{lemma}

\begin{proof}
First, we have that
\begin{equation*}
    \Exp{K_n(r;\beta)} = \frac{1}{E_n(r;\beta)}\sum_{k=0}^n k a^n_k(\beta)r^k
    = \frac{rE'_n(r,\beta)}{E_n(r;\beta)}.
\end{equation*}
Then, using the representation
\begin{equation*}
    E_n(r;\beta) = \prod_{i=1}^n \pqty{1+r\upe^{-\beta i/n}}
\end{equation*}
taking the logarithm, and differentiating with respect to $r$,
\begin{equation*}
    \frac{E'_n(r;\beta)}{E_n(r;\beta)} =
        \sum_{i=1}^n \frac{\upe^{-\beta i/n}}{1+r\upe^{-\beta i/n}}.
\end{equation*}
Scaling by $r/n$, we obtain the Riemann sum
\begin{equation*}
    \frac{rE'_n(r;\beta)}{nE_n(r;\beta)} = \frac{\Exp{K_n(r,\beta)}}{n} =
    \frac{1}{n}\sum_{i=1}^n \frac{r\upe^{-\beta i/n}}{1+r\upe^{-\beta i/n}}.
\end{equation*}
Taking the limit,
\begin{align}
\nonumber
    \lim_{n\to\infty} \frac{\Exp{K_n(r,\beta)}}{n} &=
    \int_0^1 \frac{r \upe^{-\beta x}}{1+r\upe^{-\beta x}} \dd{x} \\
\label{eq:exp-limit}
    &= \frac{1}{\beta}\ln(\frac{1+r}{1+r\upe^{-\beta}}).
\end{align}
Now, since $K_n(r,\beta)$ is log-concave and hence unimodal, the difference
between the mean and the mode is bounded by the standard deviation \cite{JR51},
i.e.,
\begin{equation*}
    \abs{\Exp{K_n(r,\beta)}-\kappa_n(r,\beta)} \leq \sqrt{3}\sigma_n.
\end{equation*}
Since $\sigma_n = \Theta(\sqrt{n})$ by \cref{lm:var},
\begin{equation*}
    \lim_{n\to\infty} \frac{\kappa_n(r,\beta)}{n} = \lim_{n\to\infty}
    \frac{\Exp{K_n(r,\beta)}}{n} =
    \frac{1}{\beta}\ln(\frac{1+r}{1+r\upe^{-\beta}}).
\end{equation*}
\end{proof}

We later apply \cref{lm:mode-limit} with $r=r_{t}$, such that
\begin{equation*}
    \lim_{n\to\infty} \frac{\kappa_n(r_{t},\beta)}{n} = t.
\end{equation*}

Consider two sequences $\pqty{j_n}_{n\in\N}$ and $\pqty{k_n}_{n\in\N}$ for
which $0\leq j_n,k_n\leq n$. The next lemma states that if the difference
between $j_n$ and $k_n$ is growing like $\ca{O}(\sqrt{n})$, then the logarithms
of the corresponding polynomials $a^n_{j_n}(\beta)$ and $a^n_{k_n}(\beta)$ are
growing apart at rate $o(n)$.

\begin{lemma}\label{lm:e-bound}
If $\abs{j_n - k_n} = \ca{O}(\sqrt{n})$, then, for all $\beta>0$,
\begin{equation*}
    \lim_{n\to\infty} \frac{1}{n}
    \abs{\ln a^n_{j_n}(\beta)-\ln a^n_{k_n}(\beta)} = 0.
\end{equation*}
\end{lemma}

\begin{proof}
We again leverage Newton's inequalities, which give
\begin{equation*}
    \frac{a^n_{k+1}(\beta)}{a^n_k(\beta)} \leq
    \pqty{\frac{k}{k+1}}\pqty{\frac{n-k}{n-k+1}}
    \frac{a^n_k(\beta)}{a^n_{k-1}(\beta)}.
\end{equation*}
Defining $R^n_k = \ln a^n_{k+1}(\beta)-\ln a^n_k(\beta)$ for $0\leq k\leq
n-1$, the sequence $R^n_k$ is thus strictly decreasing in $k$.
Considering the left endpoint,
\begin{equation*}
    R^n_0 = \ln(\frac{a^n_1(\beta)}{a^n_0(\beta)})
    = \ln(\sum_{i=1}^n \upe^{-\beta i/n})
    = \ln(\frac{1-\upe^{-\beta}}{\upe^{\beta/n}-1}),
\end{equation*}
and taking the usual Taylor expansion for $\upe^{x}$,
\begin{equation*}
    \upe^{\beta/n}-1 = \frac{\beta}{n}\pqty{1+\ca{O}(n^{-1})},
\end{equation*}
we may write $R^n_0$ as
\begin{align}
    \nonumber
    R^n_0 &= \ln(1-\upe^{-\beta}) -
    \ln(\frac{\beta}{n}\pqty{1+\ca{O}(n^{-1})}) \\
    \label{eq:rn1}
    &= \ln(n) + \ca{O}(1).
\end{align}
Now considering the right endpoint, we first observe that
\begin{equation*}
    a^n_{n-1}(\beta) = a^n_n(\beta)\sum_{i=1}^n \upe^{\beta i/n}
    = a^n_n(\beta) \upe^{\beta/n}
    \pqty{\frac{\upe^{\beta}-1}{\upe^{\beta/n}-1}}.
\end{equation*}
This gives
\begin{align}
    \nonumber
    R^n_{n-1} &= \ln(\frac{a^n_n(\beta)}{a^n_{n-1}(\beta)}) \\
    \nonumber
    &= -\ln(\pqty{1+\frac{\beta}{n}+\ca{O}(n^{-2})}
    \pqty{\frac{\ca{O}(1)}{\frac{\beta}{n}\pqty{1+\ca{O}(n^{-1})}}}) \\
    \label{eq:rnn-1}
    &= -\ln(n) + \ca{O}(1).
\end{align}
Since the sequence $R^n_k$ is strictly decreasing, \Cref{eq:rn1,eq:rnn-1}
together imply that, for all $0\leq k\leq n-1$,
\begin{equation*}
    \abs{\ln a^n_{k}(\beta) - \ln a^n_{k+1}(\beta)} \leq 2\ln(n) + \ca{O}(1).
\end{equation*}

By assumption, $j_n$ and $k_n$ differ by $\ca{O}(\sqrt{n})$, and thus
\begin{equation*}
    \abs{\ln a^n_{j_n}(\beta) - \ln a^n_{k_n}(\beta)}
    \leq \ca{O}\pqty{\sqrt{n}\ln(n)} 
\end{equation*}
Scaling by $1/n$ and taking the limit,
\begin{align*}
    \lim_{n\to\infty} \frac{1}{n}
    \abs{\ln a^n_{j_n}(\beta) - \ln a^n_{k_n}(\beta)}
    &\leq \lim_{n\to\infty}
    \ca{O}\pqty{\frac{\ln(n)}{\sqrt{n}}} \\
    &= 0.
\end{align*}
\end{proof}

We now prove the lower bound towards \cref{pr:ank}.

\begin{lemma}\label{lm:liminf}
For $t\in(0,1)$,
\begin{align*}
    \liminf_{n\to\infty} \frac{1}{n}\ln a_{\floor{tn}}^n(\beta) &\geq
    J(r_{t,\beta};\beta) - t\ln r_{t,\beta}.
\end{align*}
\end{lemma}

\begin{proof}
For all $r>0$,
\begin{align*}
    E_n(r;\beta) &\leq (n+1)\max_{k} \bqty{a^n_k(\beta)r^k} \\
    &= (n+1)\exp(\ln a^n_{\kappa_n(r,\beta)}(\beta) + \kappa_n(r,\beta)\ln
    r).
\end{align*}
Taking logarithms, fixing $r=r_{t,\beta}$, and using the abbreviation $\kappa_n
= \kappa_n(r_{t,\beta},\beta)$,
\begin{align*}
    \ln E_n(r_{t,\beta};\beta)
    &\leq \ln(n+1) +
    \ln a^n_{\kappa_n}(\beta) +
    \kappa_n \ln r_{t,\beta}  \\
    &\leq \ln(n+1) + \ln a^n_{\floor{tn}}(\beta)
    + \delta + \kappa_n \ln r_{t,\beta},
\end{align*}
where
\begin{equation*}
    \delta = \abs{\ln a^n_{\kappa_n}(\beta) - \ln a^n_{\floor{tn}}(\beta)}.
\end{equation*}
At $r=r_{t,\beta}$, \cref{lm:mode-limit} implies that $\kappa_n/n\to t$ as
$n\to\infty$ and \cref{lm:e-bound} implies that $\delta/n$
goes to 0 as $n\to\infty$. Thus,
\begin{equation*}
    \lim_{n\to\infty}\frac{1}{n}\ln E_n(r_{t,\beta};\beta) - t\ln r_{t,\beta}
    \leq \liminf_{n\to\infty}\frac{1}{n}\ln a^n_{\floor{tn}}.
\end{equation*}
The limit on the left-hand side is a Riemann sum which converges to
$J(r_{t,\beta};\beta)$, yielding the desired result.
\end{proof}

\subsection{Upper Bound}\label{sec:hard-upper}

The proof of the upper bound is much simpler than the lower bound. We can
simply appeal to saddle point bounds for $(1/n)\ln a^n_{\floor{tn}}(\beta)$
and substitute the choice of $r=r_{t}$ freely.

\begin{lemma}\label{lm:limsup}
For $t\in[0,1]$,
\begin{equation*}
    \limsup_{n\to\infty} \frac{1}{n}\ln a_{\floor{tn}}^n(\beta) \leq
    J(r_{t,\beta};\beta) - t\ln r_{t,\beta}.
\end{equation*}
\end{lemma}

\begin{proof}
The sequence $a^n_k(\beta)$ in $k$ has the generating function
\begin{equation*}
    E_n(r) = \sum_{k=0}^n a^n_k(\beta) r^k = \prod_{i=1}^n \pqty{1+r
    \upe^{-\beta i/n}}.
\end{equation*}
Since $E_n(r)$ is entire with positive coefficients, the saddle point bounds
\cite{FS09} yield, for all $r>0$,
\begin{equation}\label{eq:saddle-bound}
    \frac{1}{n}\ln a^n_k(\beta) \leq \frac{1}{n}\ln(
    \frac{E_n(r)}{r^k}).
\end{equation}
We substitute $k=\floor{tn}$ and $r=r_{t,\beta}$ into \cref{eq:saddle-bound}.
This yields
\begin{equation*}
    \frac{1}{n}\ln a^n_{\floor{tn}}(\beta) \leq \frac{1}{n}\sum_{i=1}^n
    \ln(1+r_{t,\beta} \upe^{-\beta i/n}) - \frac{\floor{tn}}{n}\ln r_{t,\beta}.
\end{equation*}
Taking the limit as $n\to\infty$ yields the desired result.
\end{proof}

Although the proof of \cref{lm:limsup} does not require $r_{t,\beta}$ to
actually be the optimizing saddle point, it is in fact optimal, at least
asymptotically. The tightest bound of the form of \cref{eq:saddle-bound} is
given by the value of $r$ for which the derivative of the right-hand side is
zero. To find this $r$, we examine
\begin{equation*}
    \dv{r}\pqty{\frac{1}{n}\ln E_n(r) - \frac{k}{n}\ln r} =
    \frac{1}{n}\sum_{i=1}^n \frac{\upe^{-\beta i/n}}{1+r \upe^{-\beta i/n}} -
    \frac{k}{nr}.
\end{equation*}
The optimal $r$ is thus given by the solution to
\begin{equation}\label{eq:r0-finite}
    \frac{1}{n}\sum_{i=1}^n \frac{\upe^{-\beta i/n}}{1+r \upe^{-\beta i/n}} =
    \frac{k}{nr}.
\end{equation}
In the limit, we can consider $k=tn$ for $t\in(0,1)$. The Riemann sum on the
left-hand side of \cref{eq:r0-finite} converges, and the
limiting optimal $r$ solves
\begin{equation}\label{eq:r0-limit}
    \int_0^1 \frac{r\upe^{-\beta x}}{1+r \upe^{-\beta x}} \dd{x} = t
\end{equation}
The solution to \cref{eq:r0-limit} is $r=r_{t,\beta}$, which may be seen by
comparison with \cref{eq:exp-limit}.

\subsection{Combining the Bounds}\label{sec:hard-final}

We now have all the tools necessary to complete the proof of
\cref{thm:scgf-hard}.

\begin{proof}[of \cref{thm:scgf-hard}]
As with the soft-decision sCGF (\cref{eq:lambda-renyi}), the hard-decision sCGF
is given by $\alpha H_{1/(1+\alpha)}(Z)$ for $\alpha\geq -1$. Using the
hard-decision PMF for the LRC (\cref{lm:hard-pmf}),
\begin{equation*}
\begin{split}
    \alpha H_{\frac{1}{1+\alpha}} (Z_n) &=
    -\sum_{i=1}^n \ln(1+\upe^{-\beta i/n}) \\
    &\qquad + (1+\alpha)\ln
    (\sum_{k=0}^n \binom{n}{k}^{\frac{\alpha}{1+\alpha}}
    a^n_k(\beta)^{\frac{1}{1+\alpha}}).
\end{split}
\end{equation*}
Scaling by $1/n$ and taking the limit as $n\to\infty$, the first sum is, as in
the soft-decision case, a Riemann sum with limit
\begin{equation*}
    \lim_{n\to\infty} -\sum_{i=1}^n \ln(1+\upe^{-\beta i/n}) = - J(1;\beta).
\end{equation*}

To handle the limit of the second sum,
\begin{equation}\label{eq:hard-target}
    (1+\alpha) \lim_{n\to\infty} \ln
    (\sum_{k=0}^n \binom{n}{k}^{\frac{\alpha}{1+\alpha}}
    a^n_k(\beta)^{\frac{1}{1+\alpha}}),
\end{equation}
we proceed as follows.
\begin{enumerate}[(1)]
\item We show that there exist continuous and appropriately well-behaved
functions $f,g:[0,1]\to\R$ such that, for sufficiently large $n$,
\begin{equation*}
\begin{split}
    \binom{n}{k}&^{\frac{\alpha}{1+\alpha}}
    a^n_k(\beta)^{\frac{1}{1+\alpha}} = \\
    &\exp\pqty{ n
    \bqty{\frac{\alpha}{1+\alpha} f\pqty{\frac{k}{n}} + \frac{1}{1+\alpha}
    g\pqty{\frac{k}{n}} } + o(n) }.
\end{split}
\end{equation*}

\item It follows that, in the limit, the sum in \cref{eq:hard-target}
behaves like a Riemann sum and \cref{eq:hard-target} is equal to
\begin{equation}\label{eq:v-limit}
    (1+\alpha)\lim_{n\to\infty}\frac{1}{n}\ln V(\alpha,n),
\end{equation}
where
\begin{equation*}
    V(\alpha,n) = \int_0^1
    \exp\pqty{ n \bqty{
    \frac{\alpha f(t)}{1+\alpha}+
    \frac{g(t)}{1+\alpha}
    } + o(n) } \dd{t}.
\end{equation*}

\item We apply Laplace's method to show that \cref{eq:v-limit} is equal to
\begin{equation*}
    \max_{t\in[0,1]} \bqty{\alpha f(t) + g(t)}.
\end{equation*}
\end{enumerate}

For sufficiently large $n$ and $t\in(0,1)$,
\begin{equation*}
    \binom{n}{tn} = \exp(n h(t) + o(n)).
\end{equation*}
By \cref{pr:ank}, for sufficiently large $n$ and $t\in(0,1)$,
\begin{equation*}
    a_{\floor{tn}}^n(\beta) =
    \exp\bqty{ n\pqty{ J\pqty{r_{t,\beta};\beta} - t\ln r_{t,\beta}} +
    o(n)}.
\end{equation*}
For both approximations, the error terms are uniform over compact subsets of
$t\in(0,1)$ and the linear terms in the exponents are continuous functions of
$t$. This suffices for the equality of \cref{eq:hard-target} and
\cref{eq:v-limit}. To apply Laplace's method to \cref{eq:v-limit}, the
function
\begin{equation}\label{eq:max-arg}
    f_{\alpha,\beta}(t) = \alpha h(t) + J\pqty{r_{t,\beta}; \beta} - t\ln
    r_{t,\beta}.
\end{equation}
must have a unique maximum over $t\in(0,1)$ and a negative second derivative
with respect to $t$ over that range; these properties are proven in
\cref{lm:hard-concave}.

Having established the behavior of $\Lambda_Z(\alpha)$ for
$\alpha\in(-1,\infty)$, we now confirm that $\Lambda_Z(\alpha)$ also has
a continuous derivative for $\alpha>-1$, which suffices to establish that
$\Lambda_Z(\alpha)=-H_{\up{min}}(Z)$ for $\alpha\leq -1$ \cite[Lemma 1]{CD12}.
For any fixed $t$, \cref{eq:max-arg} is linear in $\alpha$. Since the maximum
of linear functions is convex, $\Lambda_Z(\alpha)$ must be convex for
$\alpha>-1$. Since $h(t)$ is strictly concave, the maximizing $t$ is a
continuous function of $\alpha$ and is unique for each $\alpha>-1$. These
together imply that $\Lambda_Z(\alpha)$ has a continuous derivative and
thus that $\Lambda_Z(\alpha)=-H_{\up{min}}(Z)$ for $\alpha<-1$. Since
the unique most probable noise effect is the all-zero sequence, we again have,
as in the soft-decision setting, that $-H_{\up{min}}(Z) = -J(1;\beta)$.
\end{proof}

\section{Proof of the Ordering of Critical Rates}\label{apx:d-order}

We show here that $\Lambda'_N(1) < \Lambda'_Z(1)$, which is equivalent to
showing that the critical rate for soft-decision decoding is higher than that
for hard-decision decoding in the LRC. We first derive an alternate expression
for $\Lambda'_Z(1)$.

\begin{lemma}\label{lm:hard-diff}
Let $\Lambda_Z$ be the sCGF for hard-decision guesswork in the LRC. Then,
$\Lambda'_Z(1)=h(t(\beta))$, where $t(\beta)\in(0,1/2)$ is the maximizer
in the expression given in \cref{thm:scgf-hard} for $\Lambda_Z(1)$,
i.e., the solution to
\begin{equation}\label{eq:h-diff}
    \dv{t}\bqty{h(t) + J(r_{t,\beta};\beta) - t\ln r_{t,\beta}} = 0.
\end{equation}
Furthermore, $t(\beta)$ is the unique solution to
\begin{equation*}
    \upe^{\beta(1-t)} = 1+\pqty{\upe^{\beta}-1}t.
\end{equation*}
\end{lemma}

\begin{proof}
Overloading notation slightly, let
\begin{equation*}
    t(\alpha,\beta) =
    \argmax_{t\in[0,1]} \bqty{
    \alpha h(t) + J(r_{t,\beta};\beta) - t\ln r_{t,\beta}
    }.
\end{equation*}
By the envelope theorem \cite{Aki85},
\begin{equation*}
    \Lambda'_Z(\alpha) = h(t(\alpha,\beta)).
\end{equation*}
Thus,
\begin{equation*}
    \Lambda'_Z(1) = h(t(1,\beta)) = h(t(\beta)).
\end{equation*}

Taking the derivative in \cref{eq:h-diff} and manipulating,
\begin{equation*}
    t(\beta) =
    \frac{1}{\beta}W_0\pqty{
        \frac{\beta}{a}\upe^{\beta}\upe^{\beta/a}
    } - \frac{1}{\upe^{\beta}-1},
\end{equation*}
where $a=\upe^{\beta}-1>0$ and $W_0$ is the principal branch of the Lambert $W$
function, which solves $W_0(z)\upe^{W_0(z)}=z$ when $z$ is real and positive.
Using this property and manipulating exponentials yields
\begin{equation*}
    \upe^{\beta(1-t(\beta))} = 1+\pqty{\upe^{\beta}-1}t(\beta).
\end{equation*}

Finally, define the function
\begin{equation*}
    f_{\beta}(t) = \upe^{\beta(1-t)} - 1 - \pqty{\upe^{\beta}-1}t,
\end{equation*}
which is strictly decreasing in $t$. Since $f_{\beta}(0) > 0$ and
$f_{\beta}(1/2)<0$, it follows that $t(\beta)\in(0,1/2)$.
\end{proof}

We now establish the strict ordering $\Lambda'_N(1) < \Lambda'_Z(1)$.

\begin{lemma}\label{lm:scgf-d-order}
Let $\Lambda_N$ and $\Lambda_Z$ be the sCGFs for soft- and hard-decision
guesswork in the LRC. Then, $\Lambda'_N(1) < \Lambda'_Z(1)$.
\end{lemma}

\begin{proof}
Making explicit the LRC channel parameter $\beta>0$, let
$d_N(\beta)=\Lambda'_N(1)$ and $d_Z(\beta)=\Lambda'_Z(1)$. We show that
$d_N(\beta) < d_Z(\beta)$ for all $\beta$.

It is straightforward to verify that
\begin{equation*}
    d_N(\beta) = \frac{\pi^2}{3\beta} - \ln(1+\upe^{-\beta/2}) +
    \frac{4}{\beta}\dilog{1+\upe^{-\beta/2}},
\end{equation*}
using the dilogarithm function
\begin{equation*}
    \dilog{x} = \int_1^x \frac{\ln x}{1-x} \dd{x}.
\end{equation*}
By \cref{pr:scgf-slope}, $d_N(\beta)\in(0,\ln2)$, and thus there exists
$p(\beta)\in(0,1/2)$ such that $d_N(\beta) = h(p(\beta))$. Define, for $t>0$,
\begin{equation*}
    q(t) = \frac{1}{1+\upe^{t/2}} \in \pqty{0,1/2}.
\end{equation*}
Letting $s=\upe^{-\beta/2}$,
\begin{align*}
    \dv{\beta}\bqty{3\beta d_N(\beta)} &= 3\pqty{\ln(1+s) +
        \frac{\beta s}{2\pqty{1+s}}} \\
    &= 3h(q(\beta)).
\end{align*}
Since $3\beta d_N(\beta)$ vanishes at $\beta=0$,
\begin{equation*}
    d_N(\beta) = \frac{1}{\beta}\int_0^\beta h(q(t))\dd{t}.
\end{equation*}

Now, define the average $\overline{q}(\beta)\in (0,1/2)$,
\begin{align*}
    \overline{q}(\beta) &= \frac{1}{\beta} \int_0^\beta q(t)\dd{t} \\
    &= 1 - \frac{2}{\beta} \ln(\frac{1+\upe^{\beta/2}}{2}).
\end{align*}
Treating $t$ as a random variable uniformly distributed over $(0,\beta)$,
Jensen's inequality applied to the strictly concave function $h$ gives
\begin{equation*}
    d_N(\beta) =
    \int_0^\beta \frac{1}{\beta}h(q(t)) \dd{t} <
    h\pqty{\int_0^\beta \frac{1}{\beta}q(t) \dd{t}} = h(\overline{q}(\beta)).
\end{equation*}
As $h$ is strictly increasing on $(0,1/2)$, it follows that $p(\beta) <
\overline{q}(\beta)$.

By \cref{lm:hard-diff}, $d_Z(\beta)=h(t(\beta))$ where
$t(\beta)\in(0,1/2)$ is the unique zero of the strictly decreasing function
\begin{equation*}
    f_{\beta}(t) = \upe^{\beta(1-t)} - 1 - \pqty{\upe^{\beta}-1}t,
\end{equation*}
Letting $z=\upe^{\beta/2}>1$,
\begin{equation*}
    f_{\beta}(\overline{q}(\beta)) =
        \frac{z^2-1}{\ln z}\ln(\frac{1+z}{2}) + \frac{-3z^2+2z+1}{4}.
\end{equation*}
\Cref{lm:zpos} shows that the right-hand side is positive.
Because $f_{\beta}$ is strictly decreasing, its zero $t(\beta)$ satisfies
$\overline{q}(\beta) < t(\beta)$.
Thus, $p(\beta)<t(\beta)$, and in turn $d_N(\beta) < d_Z(\beta)$.
\end{proof}

The following result is used in the proof of \cref{lm:scgf-d-order}.

\begin{lemma}\label{lm:zpos}
For all $z>1$,
\begin{equation*}
    f(z) = \frac{z^2-1}{\ln z}\ln(\frac{1+z}{2}) - \frac{3z^2-2z-1}{4} > 0.
\end{equation*}
\end{lemma}

\begin{proof}
Factoring out $z^2-1$,
\begin{align*}
    f(z) &= (z^2-1)\bqty{w(z)-r(z)}, \\
    w(z) &= \frac{\ln(\frac{1+z}{2})}{\ln z}, \\
    r(z) &= \frac{3z+1}{4z+4}.
\end{align*}
We show that $w(z)>r(z)$ for all $z>1$. Note that
\begin{equation*}
    \lim_{z\to1}w(z) = \lim_{z\to1}r(z) = \frac{1}{2}.
\end{equation*}
It thus suffices to show that $y(z) = w(z)-r(z)$ is strictly increasing for
$z>1$. We have that
\begin{equation*}
    y'(z) = \frac{ n(z) } {2(z+1)^2\pqty{\ln z}^2 z},
\end{equation*}
where the numerator is
\begin{equation*}
\begin{split}
    n(z) &= -2(z+1)^2\ln(z+1) + 2(z+1)^2\ln 2 \\
        &\qquad+ \pqty{z\ln z}\pqty{2z+2-\ln z}.
\end{split}
\end{equation*}
The denominator of $y'(z)$ is positive for $z>1$, so it suffices to show that
$n(z)$ is also positive. To do so, we repeatedly take derivatives until arrive
at an expression which is readily shown to be positive.

The first three derivatives of $n$ are
\begin{align*}
    n'(z) &= 4(z+1)\ln(\frac{2}{z+1}) + 4z\ln z - \pqty{\ln z}^2, \\
    n''(z) &= 4\ln(\frac{2z}{z+1}) - \frac{2\ln z}{z}, \\
    n'''(z) &= \frac{2}{z^2}\pqty{\ln z + \frac{z-1}{z+1}}.
\end{align*}
Since both $\ln(z)>1$ and $(z-1)/(z+1)>0$ for $z>1$, it follows that
$n'''(z)>0$ for all $z>1$. Since
\begin{equation*}
    n''(1)=n'(1)=n(1)=0,
\end{equation*}
it follows that $n''(z)$, $n'(z)$, and $n(z)$ are all positive for $z>1$. Thus,
$y'(z)>0$ for all $z>1$, as desired.
\end{proof}

\end{document}